\documentclass[11pt]{article}
 \usepackage{graphicx}
\usepackage{tikz-cd}

\usepackage{graphics,graphicx,color,float, hyperref}
\usepackage{epsfig,bbm}
\usepackage{bm}
\usepackage{mathtools}
\usepackage{amsmath}
\usepackage{amssymb}
\usepackage{amsfonts}
\usepackage{amsthm}
\usepackage {times}
\usepackage{layout}
\usepackage{setspace}

\newtheorem{fact}{Fact}

\newcommand{\beq}{\begin{equation}}
\newcommand{\enq}{\end{equation}}
\newcommand{\bel}{\begin{lemma}}
\newcommand{\enl}{\end{lemma}}
\newcommand{\bet}{\begin{theorem}}
\newcommand{\ent}{\end{theorem}}

\newcommand{\tr}{\mathrm{Tr}}

\newcommand{\Tr}{\mathrm{Tr}}
\newcommand{\ketbra}[1]{|#1\rangle\langle#1|}

\newcommand{\eps}{\varepsilon}

\newcommand{\mmod}[1]{\hspace{1mm}\ensuremath{(\text{mod }#1)}}

\newcommand{\brc}{\ensuremath{G}}
\newcommand{\brce}{\ensuremath{F}}
\newcommand{\qbit}{\ensuremath{Q}}

\newcommand*{\cH}{\mathcal{H}}

\newcommand*{\cD}{\mathcal{D}}

\newcommand*{\cN}{\mathcal{N}}
\newcommand*{\cS}{\mathcal{S}}

\newcommand*{\cX}{\mathcal{X}}

\newcommand*{\cE}{\mathcal{E}}

\newcommand{\cP}{\mathcal{P}}

\newcommand{\supp}{\mathrm{supp}}
\newcommand{\suppress}[1]{}

\newcommand{\defeq}{\ensuremath{ \stackrel{\mathrm{def}}{=} }}
\newcommand{\F}{\mathrm{F}}
\newcommand{\Pur}{\mathrm{P}}
\newcommand {\br} [1] {\ensuremath{ \left( #1 \right) }}

\newcommand {\minusspace} {\: \! \!}
\newcommand {\smallspace} {\: \!}
\newcommand {\fn} [2] {\ensuremath{ #1 \minusspace \br{ #2 } }}
\newcommand {\ball} [2] {\fn{\mathcal{B}^{#1}}{#2}}
\newcommand {\relent} [2] {\fn{\mathrm{D}}{#1 \middle\| #2}}
\newcommand {\dmax} [2] {\fn{\mathrm{D}_{\max}}{#1 \middle\| #2}}
\newcommand {\dmaxeps} [3] {\fn{\mathrm{D}_{\max}^{#3}}{#1 \middle\| #2}}
\newcommand {\mutinf} [2] {\fn{\mathrm{I}}{#1 \smallspace : \smallspace #2}}
\newcommand {\imax}{\ensuremath{\mathrm{I}_{\max}}}
\newcommand {\imaxeps}[3]{\ensuremath{\mathrm{I}^{#3}_{\max}(#1:#2)}}

\newcommand {\hmin} [2] {\fn{\mathrm{H}_{\min}}{#1 \middle | #2}}
\newcommand {\hmineps} [3] {\fn{\mathrm{H}^{#3}_{\min}}{#1 \middle | #2}}

\newcommand {\dheps} [3] {\ensuremath{\mathrm{D}_{\mathrm{H}}^{#3}\left(#1 \| #2\right)}}
\newcommand {\id} {\ensuremath{\mathrm{I}}}

\newcommand {\chnl}[1]{\ensuremath{\cN_{A\to B}{\left(#1\right)}}}
\newcommand {\wal}{\ensuremath{W^{alice}}}
\newcommand {\wbob}{\ensuremath{W^{bob}}}

\newcommand*{\cL}{\mathcal{L}}

\newcommand{\bra}[1]{\langle #1|}
\newcommand{\ket}[1]{|#1 \rangle}

\mathchardef\mhyphen="2D

\makeatletter
\newcommand*{\rom}[1]{\expandafter\@slowromancap\romannumeral #1@}
\makeatother

\mathchardef\mhyphen="2D

\newtheorem{definition}{Definition}
\newtheorem{claim}{Claim}

\newtheorem{theorem}{Theorem}
\newtheorem{lemma}{Lemma}
\newtheorem{corollary}{Corollary}
\begin{document}
\title{Efficient methods for one-shot quantum communication}
\author{
Anurag Anshu\footnote{Institute for Quantum Computing, University of Waterloo and Perimeter Institute for Theoretical Physics. \texttt{aanshu@uwaterloo.ca}} \qquad
Rahul Jain\footnote{Centre for Quantum Technologies, National University of Singapore; MajuLab, UMI 3654, 
Singapore and VAJRA Adjunct Faculty, TIFR, Mumbai, India. \texttt{rahul@comp.nus.edu.sg}} \qquad 
} 

\date{}
\maketitle

\begin{abstract}
We address the question of efficient implementation of quantum protocols, with small communication and entanglement, and short depth circuit for encoding or decoding. We introduce two new methods to achieve this, the first method involving two new versions of the convex-split lemma that use much smaller amount of additional resource (in comparison to previous version) and the second method being inspired by the technique of classical correlated sampling in computer science literature. These lead to a series of new consequences, as follows.

First, we consider the task of quantum decoupling, where the aim is to apply an operation on a $n$-qubit register so as to make it independent of an inaccessible quantum system. Many previous works achieve decoupling with the aid of a random unitary. It is known that random unitaries can be replaced by random circuits of size $\mathcal{O}(n\log n)$ and depth $\text{poly}(\log n)$, or unitary $2$ designs based on Clifford circuits of similar size and depth. We show that given any choice of basis such as the computational basis, decoupling can be achieved by a unitary that takes basis vectors to basis vectors. Thus, the circuit acts in a `classical' manner and additionally uses $\mathcal{O}(n)$ catalytic qubits in maximally mixed quantum state. Our unitary performs addition and multiplication modulo a prime and hence achieves a circuit size of $\mathcal{O}(n\log n)$ and logarithmic depth. This shows that the circuit complexity of integer multiplication (modulo a prime) is lower bounded by the optimal circuit complexity of decoupling.

Next, we construct a new one-shot entanglement-assisted protocol for quantum channel coding that achieves near-optimal communication through a given channel. Furthermore, the number of qubits of pre-shared entanglement is exponentially smaller than that used in the previous protocol that was near-optimal in communication. We also achieve similar results for the one-shot quantum state redistribution. 
\end{abstract}

\section{Introduction}

It is hard to overstate the power of communication in today's society, which enjoys the benefits of technological advances due to telecommunication and the internet. These advances are a result of \textit{reliable} and \textit{efficient} classical communication protocols, which have been facilitated by decades of studies on data compression, error correction and physics of data transmission. As our technologies enter the quantum age, we have similarly started facing the question of how to make \textit{quantum communication} reliable and efficient. Quantum communication is central to the important tasks of quantum key distribution \cite{BennettB14, Ekert91}, the transfer of quantum states \cite{CiracZKM97} and the design of large scale quantum computers \cite{BrownKM16, MonroeRRBMDK14}.  While the proposals and experimental implementations of quantum communication have made great strides in recent years \cite{AzumaTL15, AzumaTM15, DuanLCZ01, Kimble2008, Ma2012, Liao2017}, the range of communication is still limited to about a few hundred kilometers \cite{PirandaloB2016, Ma2012, Liao2017} in ground-based experiments. Some of the key challenges are the probabilistic nature (as well as decoherence) in optics-based models \cite{PirandaloB2016, Ma2012, Takeoka2014, Pirandola2015} and fast decoherence in matter-based models \cite{PirandaloB2016, BurkardKD04}. This strongly motivates the problem of finding quantum protocols that efficiently achieve certain tasks with small communication or fight noise to reliably communicate a given amount of message.

The efficiency of a quantum communication protocol is typically captured by two quantities: the number of qubits communicated and the amount of additional resource, such as quantum entanglement, needed in the protocol. Since the foundational works of Holevo, Schumacher and Westmoreland \cite{Schumacher95, SchuW97, Holevo98}, great progress has been made in the understanding of optimal amount of communication and additional resources needed in a large family of quantum communication tasks. Well known results on quantum channel coding \cite{Holevo98, SchuW97, lloyd97, Shor02, BennettSST02, Devetak05private, HaydenHWY08}, quantum source coding \cite{Schumacher95}, quantum state merging \cite{HorodeckiOW05, HorodeckiOW07} and quantum state redistribution \cite{Devatakyard, YardD09} have discovered a powerful collection of tools for quantum information processing. These tools have found applications in disciplines beyond quantum communication, such as quantum thermodynamics \cite{LindenPSW09, RioARDV11} and black hole physics \cite{Page93, HaydenP07}. One such tool that takes a central stage in our work is that of quantum decoupling.

Notably, aforementioned works in quantum information theory are set in the asymptotic and i.i.d. (independent and identically distributed) framework of Shannon \cite{Shannon}, which allows the protocol to run over many independent instances of the input system. In practice, however, one typically does not have an access to such independent instances, limiting the scope of these results. The field of one-shot information theory addresses this problem, by constructing protocols that run on one instance of the input system. This leads to a generalization of the asymptotic and i.i.d. theory and brings information processing tasks to a more practical domain. 

However, unlike the asymptotic and i.i.d. theory of quantum information, the understanding of optimal communication and additional resources is still lacking in one-shot quantum information theory. Even for the very basic task of entanglement-assisted quantum channel coding \cite{BennettSST02}, state-of-the-art \cite{DattaH13, DattaTW2016, AnshuJW17CC} one-shot protocols fail to simultaneously achieve optimal communication capacity and optimal amount of initial entanglement. The aim of this work is to introduce new methods that make progress in this problem and exponentially improve upon the amount of initial entanglement needed in a family of one-shot protocols that achieve the best known communication for above tasks. In many cases, the resulting protocols have the additional property that either the encoding or the decoding operation is a quantum circuit of small depth. 

In order to lay the groundwork for our results, we revisit the existing techniques of decoupling and more recent convex-split and position-based decoding. Decoupling (see Figure \ref{decoupling}) refers to the process of applying some quantum operation on one of the two given systems (which share quantum correlation), so as to make the two systems independent of each other. This idea has been applied in the aforementioned tasks of quantum state merging \cite{HorodeckiOW05, HorodeckiOW07, ADHW09, Berta09, Renner11}, quantum state redistribution \cite{Devatakyard, YardD09, DattaHO16, BertaCT16} and quantum channel coding \cite{Devetak05private, DupuisHL10, DattaH13, DattaTW2016}, as well as randomness extraction \cite{Renner05, Berta13, BertaFW14}. The central approach in many of these works is to perform a random unitary operation \cite{HorodeckiOW05, HorodeckiOW07} and then discard a part of the system. This technique has been expanded upon in various works such as \cite{Frederic10, Szehr11, DupuisBWR14}. Due to the importance of decoupling technique and the limitation that random unitaries cannot be implemented with a quantum circuit of small size, there is a great interest in finding efficient circuits that achieve the same performance as a random unitary. 

Existing methods to make decoupling efficient involve replacing random unitaries with unitary 2-designs \cite{DankertCEL09, DivincenzoLT02, Chau05, CleveLLC16} which can be simulated by Clifford circuits of small depth, random quantum circuits of small depth \cite{BrownF15} and random unitaries diagonal in Pauli-$\mathsf{X}$ and Pauli-$\mathsf{Z}$ basis \cite{NakataHMW17}. To elaborate, suppose we are given a quantum state $\Psi_{RC}$ on two registers $R$ and $C$, and we need to make $C$ independent of $R$ by acting on $C$. We must further ensure that the size of the discarded system, which is the cost of the decoupling operation (see Figure \ref{decoupling}), is small enough \footnote{The number of qubits of the discarded system translates to the quantum communication cost of a quantum protocol that employs decoupling. This motivates the question of minimizing the size of discarded system.}, ruling out the operation that discards all of $C$.  The work \cite{CleveLLC16} shows that a quantum circuit of size $\mathcal{O}(\log|C|\log\log|C|)$ and depth $\mathcal{O}(\log\log|C|)$ suffices for this purpose, achieving the same cost as that of a random unitary. A similar circuit size of $\mathcal{O}(\log|C|\log^2\log|C|)$ and depth $\mathcal{O}(\log^3\log|C|)$ is obtained in \cite{BrownF15}, using elementary gates that mimick real world quantum processes. 
 
While the circuit size achieved by above results is impressive, the gates used in the circuit are highly quantum. More precisely, for a choice of preferred basis such as the computational basis, the gates convert any basis vector into a superposition over these vectors. Can the construction of a decoupling operation be further simplified, by only using the gates that are classical (taking basis vectors to basis vectors)? While being useful for practical implementation, such a construction would also lead to a surprising theoretical simplification: it would leave no conceptual difference between quantum decoupling and its classical counterpart of randomness extraction \cite{NisanZ96, RadhakrishnanT00, Trevisan01}.

\begin{figure}[!h]
\center
\includegraphics[width=10cm]{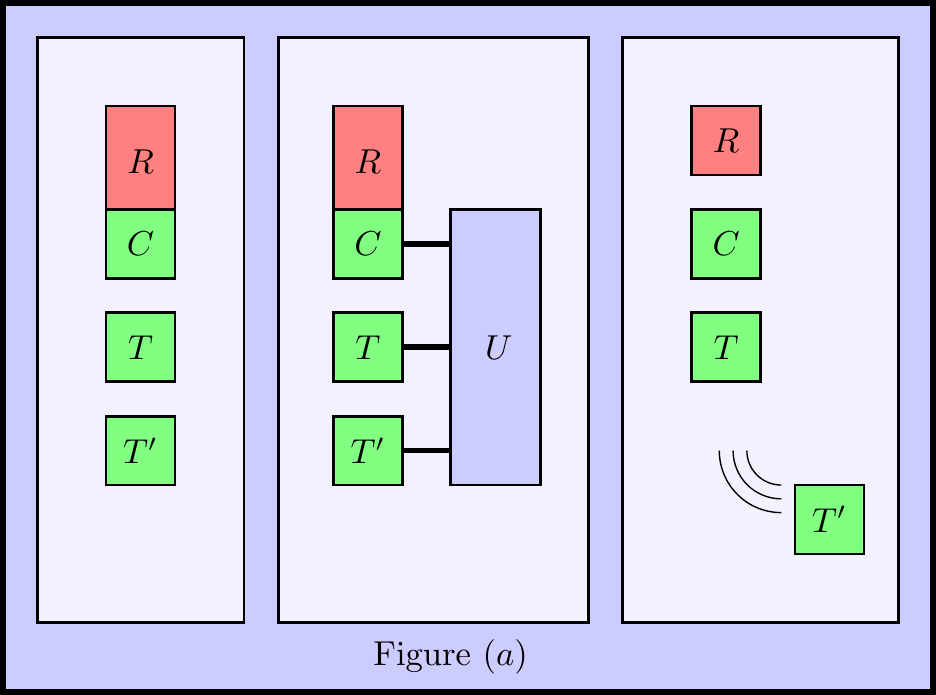}
\\
\includegraphics[width=10cm]{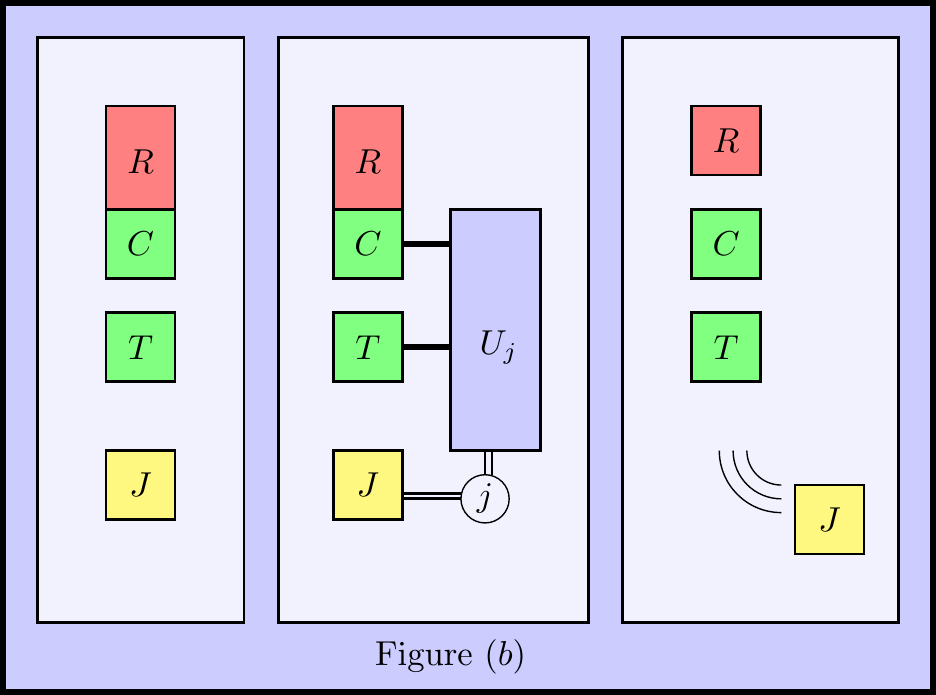}
\caption{{\small Decoupling method refers to removing the quantum correlation between two registers $R$ and $C$, by means of quantum operations. The cost of performing a decoupling operation is characterized by the size of the register that must be discarded, in order to implement the operation. In $a)$, the discarded register is $T'$ and the operation performed on $CTT'$ is a global unitary $U$. In $b)$, the register $J$ (that is eventually discarded) is maximally mixed to begin with and the operation performed is a controlled unitary. Thus, $J$ can be viewed as a classical noise \cite{GroismanPW05}. While the operation in $b)$ is a special kind of operation in $a)$, the following equivalence holds due to the duality between teleportation \cite{Teleportation93} and superdense coding \cite{BennettW92}. For every operation in $a)$ with $\log|T'|$ qubits that are discarded, there is an operation in $b)$ with $2\log|T'|$ bits of noise. Moreover, for every operation in $b)$ with $\log|J|$ bits of noise, there is an operation in $a)$ where $\frac{1}{2}\log|J|$ qubits that are discarded.}}
\label{decoupling}
\end{figure}

Random permutation is a canonical classical operation known to perform randomness extraction and also decouple classical-quantum systems \cite{Renner05, Berta13, BertaFW14}. In \cite{DupuisDT14} (see also \cite{Szehr11}) the authors used permutations to derive an analogue of the decoupling theorem that however only removes quantum and not classical correlations between $R$ and $C$. While the remaining classical correlation could also be removed by random permutations, the overall cost of decoupling would be larger than the cost of decoupling by a random unitary. This indicates that a decoupling method, which matches the random unitary decoupling in its cost, can only involve operations that are not classical. 

This is shown not to be true by the convex-split lemma \cite{AnshuDJ14}, which expresses a relation of the following form
\begin{equation}
\label{convsplit}
\Phi_{RCE} \approx \sum_i p_i \Phi^{(i)}_{RCE},
\end{equation} 
showing how to view a given quantum state $\Phi_{RCE}$ as a convex combination of (more desirable) quantum states $\Phi^{(i)}_{RCE}$ in order to achieve an information-theoretic task. It implies decoupling (of the type in Figure \ref{decoupling}, $(b)$) when the quantum state on the left hand side (that is, $\Phi_{RCE}$) is a product state across $R$ and $CE$. In particular, it was shown in \cite{AnshuDJ14} that given $\Psi_{RC}$, if we add the quantum state $\sigma_{C_1}\otimes \ldots \sigma_{C_N}$ (for some large enough $N$) and randomly swap the register $C$ with one of the registers $C_1, \ldots C_N$, then the register $R$ becomes independent of all the other registers \footnote{Expressed mathematically via Equation \ref{convsplit}, we set $E=C_1 C_2\ldots C_N$, $\Phi_{RCE}= \Psi_R\otimes \sigma_C\otimes \sigma_{C_1}\otimes \ldots \sigma_{C_N}$, $\Phi^{(i)}_{RCE}=\Psi_{RC_i}\otimes \sigma_C\otimes \sigma_{C_1}\otimes\ldots \sigma_{C_{i-1}}\otimes \sigma_{C_{i+1}}\otimes\ldots \sigma_{C_N}$ and $p_i=\frac{1}{N}$.}; leading to decoupling with the classical operation of permutation of registers. In this work we will solely be interested in quantum tasks where decoupling is the same as constructing an appropriate convex-split, and hence we will use the two terms interchangeably. However, we highlight that the convex-split method is more general and can be used even in situations where no decoupling exists: such as in classical or classical-quantum communication tasks \cite{AnshuJW17CC, AnshuJW17MC, AGHY16} and resource theoretic tasks \cite{AnshuJH18, BertaM18, LiuW19}. 

Since the process of swapping two registers is a `classical' operation (that is, it takes basis vectors to basis vectors), the convex-split lemma of \cite{AnshuDJ14} gives a classical unitary for performing quantum decoupling. Unfortunately, the value of $N$ can be as large as $\mathcal{O}(|C|)$, where $|C|$ is the dimension of the register $C$. Hence swapping the register $C$ with a random register $C_i$ requires a circuit of depth $\mathcal{O}(|C|)$, which is exponential in the number of qubits of register $C$. Even an alternate implementation of swap operation, by placing the registers on a three dimensional grid, would require $\mathcal{O}(|C|^{1/3})$ operations. Thus, it has so far been unknown if one can achieve quantum decoupling by efficient classical operations. 

Recent works have shown several applications of the convex-split method in one-shot quantum information theory, along with the dual method of position-based decoding \cite{AnshuJW17CC}. The methods have been used to obtain near-optimal communication for one-shot entanglement-assisted quantum channel coding \cite{AnshuJW17CC}, near-optimal communication for one-shot quantum state splitting \cite{AnshuDJ14} (with slight improvement of the additive $\log\log|C|$ factor over \cite{Renner11}, for communicating the register $C$) and smallest known communication for one-shot quantum state redistribution \cite{AnshuJW17SR}. As mentioned earlier, all these protocols use a large amount of entanglement. Other known protocols, \cite{BennettSST02, DattaH13, DattaTW2016} for entanglement-assisted quantum channel coding and \cite{BertaCT16, DattaHO16} for quantum state redistribution, that do not rely on these two methods use exponentially small entanglement, but their communication is not known to be near-optimal. This motivates the question of find a scheme that achieves the best of both of the lines of work. 
\vspace{0.1in}

\section{Our results}

We show how to achieve near-optimal communication and the size of initial entanglement at most constant factors away from the optimal, in all the aforementioned quantum communication tasks. We further show that, in several cases, the implementation of either the encoding or the decoding operation in the protocol can be made efficient. Our results are obtained by two new methods that we outline below.   

\vspace{0.1in}

\noindent{\bf Efficient decoupling procedures (Method $A$):} As mentioned earlier, the quantity of interest in a decoupling procedure is the number of bits or qubits that are discarded to achieve the decoupling. There are two models under which decoupling is performed, see Figure \ref{decoupling}.  The first model involves adding a quantum state, applying a global unitary (without involving the register $R$) and then discarding some quantum system. The second model also involves adding a quantum state followed by a unitary, but the system that is discarded is classical and the unitary acts in a classical-quantum manner \cite{GroismanPW05}. The two models can be converted into each other by a Clifford circuit of depth $1$ and the number of qubits/bits discarded are the same up to a factor of $2$, due to the well known duality between teleportation \cite{Teleportation93} and super-dense coding \cite{BennettW92}. Additional quantum systems that are not discarded act as a catalyst for the decoupling process \cite{Renner11, AnshuDJ14, MajenzBDRC17, AnshuJH18, BertaM18}. For example, the randomness used in the process of decoupling via unitary $2$-design acts as a catalyst. In principle, this randomness can be fixed by standard derandomization arguments, but it leads to a loss in efficient implementation. 

 In this work, we consider the second model of decoupling. We construct two new convex-split lemmas which immediately lead to efficient decoupling procedures for a quantum state $\Psi_{RC}$ (recall the discussion following Equation \ref{convsplit}). One of these lemmas solves the aforementioned problem of decoupling via an efficient classical operation.
\begin{itemize}
\item {\bf Method $A.1$:} A set of unitaries $\{V_{\ell}\}_{\ell=1}^{|C|^2}$ on a register $C$ forms a $1$-design if 
$$\frac{1}{N}\sum_{\ell}V_{\ell}\rho_C V^{\dagger}_{\ell}= \frac{\id_C}{|C|}, \quad \forall \text{ quantum state } \rho_C.$$
A canonical example of unitary $1$-design is $\cP_{\log|C|}$, the set of the tensor products of Pauli $\mathsf{X}$ and $\mathsf{Z}$ operators if the register $C$ admits a qubit decomposition. Our first procedure shows how to achieve decoupling using a mixture of small number of $\approx \log|C| - \hmin{C}{R}_{\Psi}$ unitaries from any $1$-design. Here $\Psi_{RC}$ is the quantum state on registers $R$ and $C$ and $\hmin{C}{R}$ is the conditional min-entropy. The additional randomness used to choose the unitaries is $4\log|C|$ bits. We highlight that this is in stark contrast with many of the previous constructions for decoupling, which required unitaries from a $2$-design. Details appear in Subsection \ref{subsec:1design}.

\item {\bf Method $A.2$:} The second decoupling procedure enlarges the Hilbert space $\cH_C\otimes \cH_C$ in a manner that the resulting Hilbert space $\cH_{\brc}$ has prime dimension $|\brc|\leq 2|C|^2$. This is possible due to Bertrand's postulate \cite{Chebysev1852}, which says that there is a prime between any natural number and its twice. It also introduces a register $L$ of size approximately $N\defeq \log|C| - \hmin{C}{R}_\Psi$. A preferred basis on $\cH_C$ (such as the computational basis in the qubit representation of the registers) is chosen, which gives a basis $\{\ket{i}_G\}_{i=0}^{|G|-1}$ on $\cH_G$. Similarly, a preferred basis $\{\ket{\ell}\}_{\ell=1}^N$ is chosen on $\cH_L$. Following this, a unitary operation $U=\sum_{\ell=1}^NU_\ell\otimes \ketbra{\ell}_L$ is applied, where $U_\ell$ acts on two registers $\brc, \brc'\equiv \brc$ as 
\begin{equation}
\label{Uellunits}
U_\ell\ket{i}_{\brc}\ket{j}_{\brc'} = \ket{i+(j-i)\ell \mmod{|\brc|}}_{\brc}\ket{j+(j-i)\ell  \mmod{|\brc|}}_{\brc'}.
\end{equation}
Upon tracing out register $L$, register $R$ becomes independent of $\brc\brc'$. Furthermore, the final state on registers $\brc\brc'$ is maximally mixed and the register $\brc'$ is returned in the original state. As can be seen, the unitaries $U_\ell$ are `classical' as they take basis vectors to basis vectors and perform addition and multiplication modulo $|\brc|$. This makes the construction of $U$ efficient, with circuit depth $\mathcal{O}(\log\log|C|)$ and size $\mathcal{O}(\log|C|\log\log|C|)$ due to well known results in modular arithmetic \cite{McLaughlin04}. Details appear in Subsections \ref{subsec:classicalunit} (proof of decoupling) and \ref{unitimp} (circuit complexity).

In the other direction, our result shows that the reversible or quantum circuit complexity (such as depth or size) of integer multiplication modulo a prime is lower bounded by the reversible or quantum circuit complexity of the `best' decoupling method. This holds since integer multiplication is the most expensive step in Equation \ref{Uellunits}. We highlight that a super-linear lower bound on the circuit complexity of integer multiplication is an outstanding open question in the area of complexity theory \cite{SchonS71, Furer09}. The aforementioned connection to decoupling may suggest attacking this problem using an entirely different avenue connected to decoupling \cite{HaydenP07}: scrambling of quantum information in black holes \cite{LashkariSHOH13}.  
\end{itemize}

\vspace{0.1in}

\noindent{\bf Exponential improvement in entanglement (Method $B$) :} A \textit{flattening} procedure, that realizes any classical distribution as a marginal of a uniform distribution in a larger space, has been used in the context of classical correlated sampling in several works \cite{Broder97, Charikar2002, KleinbergT02, Holenstein2007, BarakHHRRS08, BravermanRao11, AnshuJW17classical}. A counterpart of this procedure for quantum states was considered in \cite{AJMSY16}. Let the eigendecomposition of $\sigma_C$ be $\sigma_C=\sum_i p_i \ketbra{i}_C$. Append a new register $E$ through the transformation
$$\ketbra{i}_C\rightarrow \ketbra{i}_C\otimes\left(\frac{1}{Kp_i}\sum_{j=1}^{Kp_i}\ketbra{j}_E\right),$$
where $K$ is a large enough real such that $\{Kp_i\}_i$ are all integers \footnote{The existence of such a $K$ can be ensured, for example, by an arbitrarily small perturbation in $\{p_i\}_i$, so that they all are rationals.}. As a result, the quantum state $\sigma_C$ transforms to  
\begin{equation}
\label{flatext}
\sigma_C\rightarrow \frac{1}{K}\sum_{i,j: j\leq Kp_i} \ketbra{i}_C\otimes \ketbra{j}_E,
\end{equation}
which is uniform in a subspace. However, \cite{AJMSY16} did not provide a unitary operation to realize the above extension of $\sigma_C$. We show that this extension can be constructed in a unitary manner using embezzling states \cite{DamH03}. If the basis $\{\ket{i}\}_i$ can be efficiently prepared from computational basis and the eigenvalues $\{p_i\}_i$ are easy to compute, then the flattening procedure is also computationally efficient. Details appear in Section \ref{sec:maxmutdec}. The consequences of this method are as follows, with all the tasks appearing below summarized in Figure \ref{qcomtasks}.

\begin{figure}[!h]
\center
\includegraphics[width=12cm]{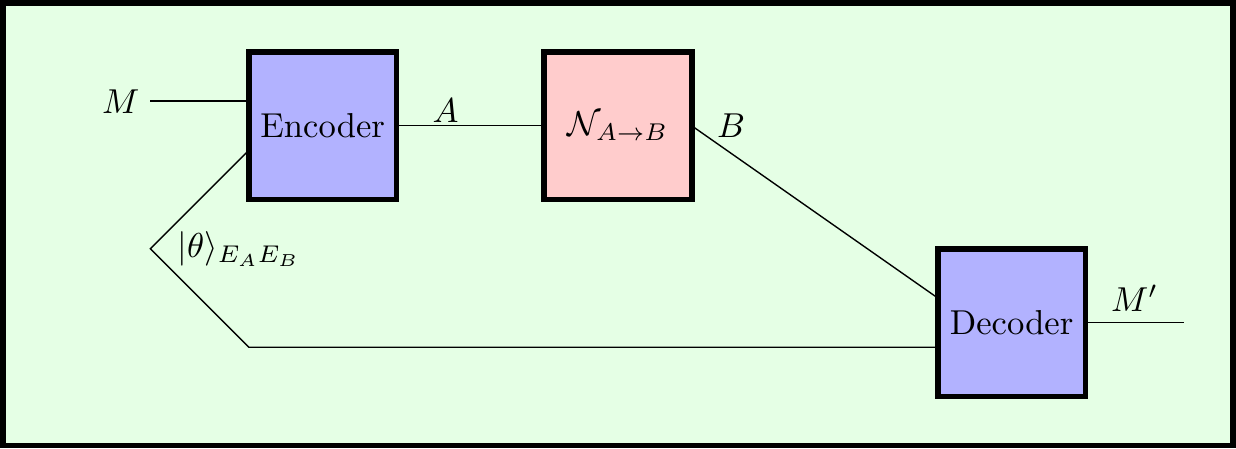}
\\
\includegraphics[width=12cm]{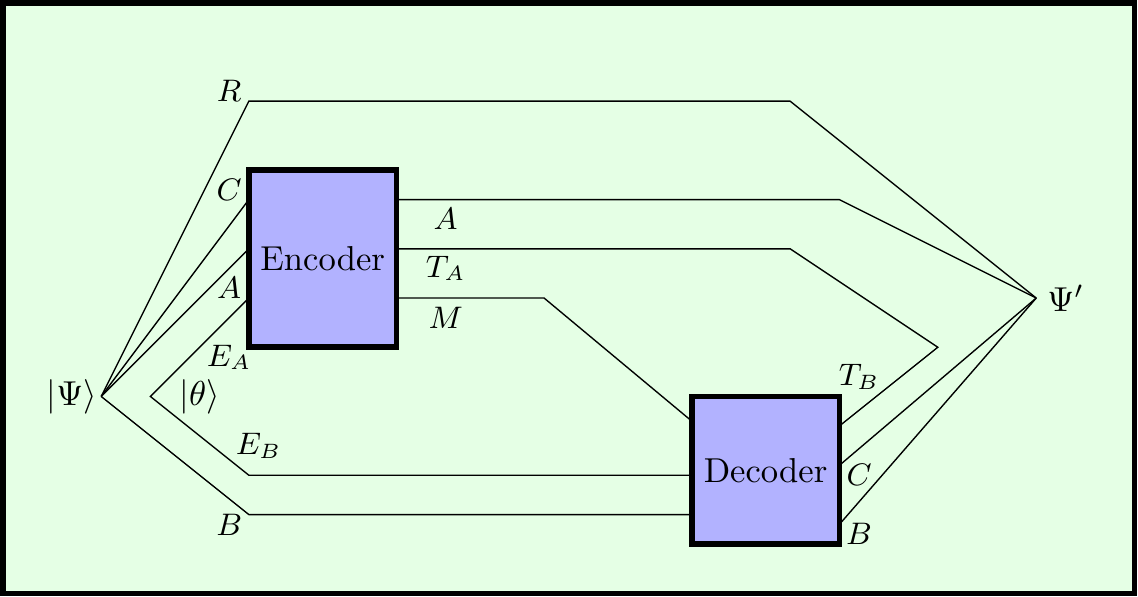}
\caption{{\small The first figure depicts the task of entanglement-assisted quantum channel coding, where the register $M$ holds a message $m\in \{1,2, \ldots 2^R\}$. The goal is to maximize the value of $R$, while keeping the error in decoding small. The second figure shows the task of quantum state redistribution with entanglement assistance. The goal is to ensure that the register $C$ is obtained by Bob using as less communication $\log|M|$ as possible and ensuring that $\Psi'\approx \ketbra{\Psi}$ .}}
\label{qcomtasks}
\end{figure}

\begin{itemize}
\item {\bf Entanglement-assisted classical communication over quantum channel:} Consider a quantum channel $\cN_{A\to B}$, over which we wish to communicate a message from the set $\{1,2,\ldots 2^R\}$, with small error. The work \cite{BennettSST02} considered the asymptotic and i.i.d. setting for this task, involving the channel $\cN_{A\to B}^{\otimes n}$ for large enough $n$. It was shown that the rate of communication $\frac{R}{n}$ converges to $$\max_{\ket{\Psi}_{AA'}}\mutinf{A'}{B}_{\cN_{A\to B}(\Psi_{AA'})},$$ where $\mutinf{A'}{B}$ is the quantum mutual information. The number of qubits of entanglement in the protocol from \cite{BennettSST02} was approximately $nS(\Psi_A)$ (the von-Neumann entropy) and the rate of communication was shown to be optimal. The work \cite{DattaTW2016} obtained a one-shot version of their protocol, with $\log |A|$ qubits of pre-shared entanglement. Their communication was characterized by the \textit{quantum hypothesis testing relative entropy} between the quantum state $\cN_{A\to B}(\Psi_{AA'})$ and a separable state derived from $\Psi_{AA'}$, which may not be optimal. The work \cite{AnshuJW17CC} introduced the position-based decoding method, showing how to achieve a communication characterized by the quantum hypothesis testing relative entropy between $\cN_{A\to B}(\Psi_{AA'})$ and $\cN_{A\to B}(\Psi_{A})\otimes \Psi_{A'}$. The achievable communication is near-optimal, due to the converse given in \cite{MatthewsW14}. But the protocol in \cite{AnshuJW17CC} required $\mathcal{O}(|A|)$ qubits of entanglement. Using our flattening procedure on the quantum state $\ket{\Psi}_{AA'}$, we show how to achieve the same near-optimal communication with $\mathcal{O}(\log|A|)$ qubits of entanglement. If the flattening procedure is efficient, then the encoding by Alice is efficient as well. Details appear in Subsection \ref{subsec:chancode}. 

The work \cite{AnshuJW17CC} also studied entanglement-assisted classical communication through various quantum networks, shown to be near optimal in \cite{AnshuJW19}. Our technique also exponentially improves upon the amount of entanglement in these protocols, while maintaining the achievable communication.
 
\item {\bf Quantum state splitting and quantum state redistribution:} The task of quantum state redistribution \cite{Devatakyard, YardD09} considers a quantum state $\ket{\Psi}_{RABC}$, where the register $R$ is inaccessible, registers $A,C$ are with Alice and register $B$ is with Bob. It is required that after communication from Alice to Bob, the register $C$ should be held by Bob. Its special cases of quantum state splitting \cite{ADHW09} and quantum state merging \cite{HorodeckiOW05} are equivalent (up to reversal of the protocol) and quantum state splitting considers the case where register $B$ is trivial. The work \cite{Renner11} obtained a one-shot protocol for quantum state splitting achieving near-optimal communication up to an additive factor of $\mathcal{O}(\log\log|C|)$. This was improved in \cite{AnshuDJ14} through a near-optimal protocol with communication tight up to an additive factor of $\mathcal{O}(1)$. While the protocol in \cite{Renner11} required $\mathcal{O}(\log|C|)$ qubits of pe-shared entanglement, the protocol in \cite{AnshuDJ14} required much larger $\mathcal{O}(|C|)$ qubits. Here, we show how to improve the number of qubits of pre-shared entanglement to $\mathcal{O}(\log|C|)$, retaining the communication cost in \cite{AnshuDJ14}. Again, we use the flattening procedure, efficiency of which ensures the efficiency of decoding operation by Bob.

The work \cite{AnshuJW17SR} gave a protocol for quantum state redistribution with smallest known quantum communication, improving upon the prior work \cite{BertaCT16}. But the number of qubits of pre-shared entanglement required was exponentially larger than that in \cite{BertaCT16}. Similar to aforementioned results, here we give a protocol that has similar quantum communication to \cite{AnshuJW17SR} and similar number of qubits of entanglement to \cite{BertaCT16}. Details appear in Subsection \ref{subsec:stateredist}.
\end{itemize}

\section{Proof outline}

The proofs of results presented in Method $A$ crucially rely on the following simple identity, which was first shown in \cite{AnshuDJ14}. Below, $\relent{.}{.}$ is the quantum relative entropy \cite{umegaki1954}.
$$\relent{\sum_i p_i \rho_i}{\theta} = \sum_i p_i \left(\relent{\rho_i}{\theta} - \relent{\rho_i}{\rho}\right).$$
This relation allows us to decompose the convex combination in Equation \ref{convsplit} into individual components. In addition, the proof of the decoupling result in Method $A.1$ also uses the notion of pairwise independent random variables to reduce the size of additional randomness, inspired by \cite{AnshuJW17MC}. The proof of decoupling result in Method $A.2$ is more subtle, as it requires us to find a collection of unitaries that form an appropriate representation of the cyclic group. Our construction, that is based on modular arithmetic, is inspired by explicit constructions of pairwise independent random variables \cite{Lovettnotes, KCN13}. 

To implement the flattening procedure in Method $B$, we show new relationships for quantum embezzlement. Let $\xi_D\defeq \frac{1}{S}\sum_{j=1}^n\frac{1}{j}\ketbra{j}_D$ be the marginal of the embezzling state from \cite{DamH03}, for some integer $n$ and $S$ being the normalization factor. Let $\rho_E\defeq \frac{1}{b}\sum_{e=1}^b\ketbra{e}_E$ be uniform in a support of size $b$. We show the existence of a unitary $U_b$ such that 
$$\dmax{U_b\left(\xi_D\otimes \ketbra{1}_E\right)U^{\dagger}_b}{\xi_D \otimes \rho_E} \leq \delta,$$ whenever $n> b^{\frac{1}{\delta}}$. Here $\dmax{.}{.}$ is the quantum max-relative entropy \cite{Datta09, Jain:2009}. Thus, it is possible to embezzle certain states with error guarantee in max-relative entropy, improving upon the error guarantee in fidelity \cite{DamH03}. We crucially use this in our proofs, as small max-relative entropy allows us to bound other one-shot information theoretic terms.

\section{Discussion}

Method $A.1$ is reminiscent of the derandomizing unitaries constructed in \cite{AmbainisS04}, which also uses unitary $1$-design for quantum encryption. But there is a difference between our setting and that in \cite{AmbainisS04}, since the number of unitaries that we use is dependent on the conditional min-entropy of the quantum state. On the other hand, the authors of \cite{AmbainisS04} only aim to decouple the maximally entangled state. We may also compare Method $A.1$ with the unitaries in \cite{NakataHMW17}, which shows how to perform decoupling with random unitaries diagonal in either $\mathsf{X}$ or $\mathsf{Z}$ bases. Our construction also yields a unitary diagonal in either $\mathsf{X}$ or $\mathsf{Z}$ bases, but it is explicit (that is, not a random unitary) and uses some additional catalytic randomness.

As mentioned earlier, the construction in Method $A.2$ is efficient, with circuit depth $\mathcal{O}(\log\log|C|)$ and size $\mathcal{O}(\log|C|\log\log|C|)$. This already achieves the performance of circuits based on unitary $2$-designs \cite{CleveLLC16} and improves upon the performance of \cite{BrownF15}, with arguably simpler construction. The unitaries $\{U_\ell\}_{\ell}$, as defined in Equation \ref{Uellunits} have an interesting property that they act as a representation of the cyclic group, reflecting the property of permutation operations in the convex-split method. 

In the language of resource theory of coherence, both the decoupling procedures in Method $A$ belong to the class of Physically Incoherent Operations \cite{StreltsovAP17}. Thus, an immediate implication of our results is that quantum decoupling can be performed by incoherent unitaries.  These decoupling procedures perform the same as decoupling via random unitary \cite{Frederic10, Berta13, DupuisBWR14}, when we consider the size of discarded system. None of these results (those in Method $A$ and the decoupling via random unitary) are optimal due to the additional effort put in making the decoupled register $C$ uniform. Indeed, it is known that the optimum cost of decoupling is characterized by the max-mutual information, rather than the conditional min-entropy \cite{Renner11, AnshuDJ14, MajenzBDRC17}. Method $B$ leads to a decoupling procedure achieving this, as it reduces the task to the case of uniform (or flat) marginal.

As shown in Equation \ref{flatext}, the central idea behind Method $B$ is to flatten a non-uniform quantum state, and use resource efficient protocols for the flattened state. The work \cite{Renner11} used a different technique for flattening the eigenvalues of a quantum state. Their technique was to distribute the eigenvalues into bins $[2^{-i}: 2^{-i-1}]$ and run a protocol within each bin (on a high level, the protocols in \cite{BennettSST02, DattaTW2016} also place the eigenvalues into uniform bins). While this method can be used for quantum state splitting (with a loss of communication of $\approx \log\log|C|$ required in transmitting the information about the bin), it is not clear how it can be used to construct a near-optimal entanglement-assisted protocol for quantum channel coding or quantum state redistribution. Our method does not face this limitation and can be uniformly applied to all the quantum communication scenarios. Further, our use of embezzling states in both quantum state splitting and entanglement-assisted quantum channel coding further highlights the duality between the two tasks \cite{BennettDHSW14, Renner11}.

We end this section with some open questions. Our first question is if there exists an analogue of Method $B$ that does not require embezzling states to achieve near-optimal decoupling. An efficient scheme could lead to new protocols with even smaller number of qubits of pre-shared entanglement in quantum communication tasks. Another important question is to see if the number of bits of additional randomness used in Method $A$ can further be reduced. It is known that seed size in randomness extraction in the presence of quantum side information can be very small \cite{DePVR12} (based on Trevisan's construction \cite{Trevisan01}). Since our construction treats classical side information and quantum side information in similar manner, we can hope to have similar results even in the case of quantum decoupling.

\subsection*{Acknowledgment} 

This work was completed when A.A. was at the Centre for Quantum Technologies, National University of Singapore, Singapore. This work is supported by the Singapore Ministry of Education through the Tier 3 Grant ``Random numbers from quantum processes'' MOE2012-T3-1-009 and VAJRA Grant, Department of Science and Technology, Government of India.

\bibliographystyle{naturemag}
\bibliography{References}

\newpage

\appendix

Here, we provide complete proofs for all the claims made in the main text. For the ease of navigation, we have discussed all the results and their interconnections in Figure \ref{fig:allresults}.

\begin{figure}[ht]
\centering
\begin{tikzpicture}

\draw[ultra thick, fill=gray!5!white] (0,0) rectangle (14,11);

\draw[thick, fill=blue!20!white] (1,10) rectangle (4,7);
\node at (2.5, 9.5) {Decoupling with};
\node at (2.5, 9) {$1$-designs};
\node at (2.5, 8.5) {``$\log|C| - $};
\node at (2.5, 8) {$\hmin{C}{R}$''};
\node at (2.5, 7.5) {Theorem \ref{paulisplit}};

\draw[thick, fill=blue!20!white] (10,10) rectangle (13,7);
\node at (11.5, 9.5) {Decoupling with};
\node at (11.5, 9) {classical ops.};
\node at (11.5, 8.5) {``$\log|C| - $};
\node at (11.5, 8) {$\hmin{C}{R}$''};
\node at (11.5, 7.5) {Theorem \ref{main:theo}};

\draw[thick, fill=yellow!30!white] (10,3.5) rectangle (13,1);
\node at (11.5, 3) {Hypothesis testing};
\node at (11.5, 2.5) {analogues of};
\node at (11.5, 2) {Theorems \ref{main:theo}, \ref{theo:maxmut} in};
\node at (11.5, 1.5) {Theorems \ref{theo:posbased}, \ref{cor:posbased}};

\draw[thick, fill=blue!20!white] (3.9,6) rectangle (7.1,3.5);
\node at (5.5,5.5) {Flattening method};
\node at (5.5,5) {derived from};
\node at (5.5,4.5) {correlated sampling:};
\node at (5.5,4) {Definitions \ref{broextend}, \ref{unitaryflat}};

\draw [->] (6.5,6) -- (6.5,7.5);
\draw[thick, fill=green!35!white] (5.8, 9.5) rectangle (8.2,7.5);
\node at (7, 9) {Decoupling};
\node at (7, 8.5) {up to};
\node at (7, 8) {``$\imax(R:C)$''};

\draw [->] (4, 8.5) -- (5.8,8.5);
\node at (4.9, 8.8) {Theorem};
\node at (4.9, 8.2) {\ref{theo:paulimaxmut}};

\draw [->] (10, 8.5) -- (8.2, 8.5);
\node at (9.1, 8.8) {Theorem};
\node at (9.1, 8.2) {\ref{theo:maxmut}};

\draw[thick, fill=red!20!white] (9,6.3) rectangle (13,4.3);
\node at (11, 5.8) {Quantum state merging};
\node at (11, 5.3) {and redistribution.};
\node at (11, 4.8) {Corollaries \ref{cor:stateredist}, \ref{cor:statemerge}};

\draw [->] (8.2, 7.5) -- (9, 6.3);
\draw [->] (11.5, 3.5) -- (11.5, 4.3);
\draw [->] (7.1, 5.3) -- (9, 5.3);

\draw[thick, fill=yellow!30!white] (1, 6) rectangle (3,4);
\node at (2, 5.5) {Hypothesis};
\node at (2, 5) {testing};
\node at (2, 4.5) {analogue};

\draw [->] (2,7) -- (2,6);
\draw [->] (2,4) -- (2,2.5);
\draw [->] (5,3.5) -- (5,2.5);

\draw [thick, fill=red!20!white] (1,2.5) rectangle (6, 1);
\node at (3.5, 2) {Entanglement-assisted quantum};
\node at (3.5, 1.5) {channel coding. Theorem \ref{optchancode}};

\end{tikzpicture}
\caption{\small An outline of our results, which are all derived in the one-shot setting. All the decoupling statements are stated as convex-split theorems. The results in blue rectangles are main tools that may be of independent interest. Method $A$ in the main text corresponds to the top two blue rectangles and Method $B$ in the main text corresponds to the lower blue rectangle. The results in red rectangles are quantum communication tasks for which we obtain entanglement cost proportional to the \textit{number of qubits of register to be communicated}, while maintaining the best known communication bounds. The result in green rectangle is the near optimal decoupling result and those in yellow rectangles are the hypothesis testing/position-based decoding analogues of convex-split theorems.}
 \label{fig:allresults}
\end{figure}
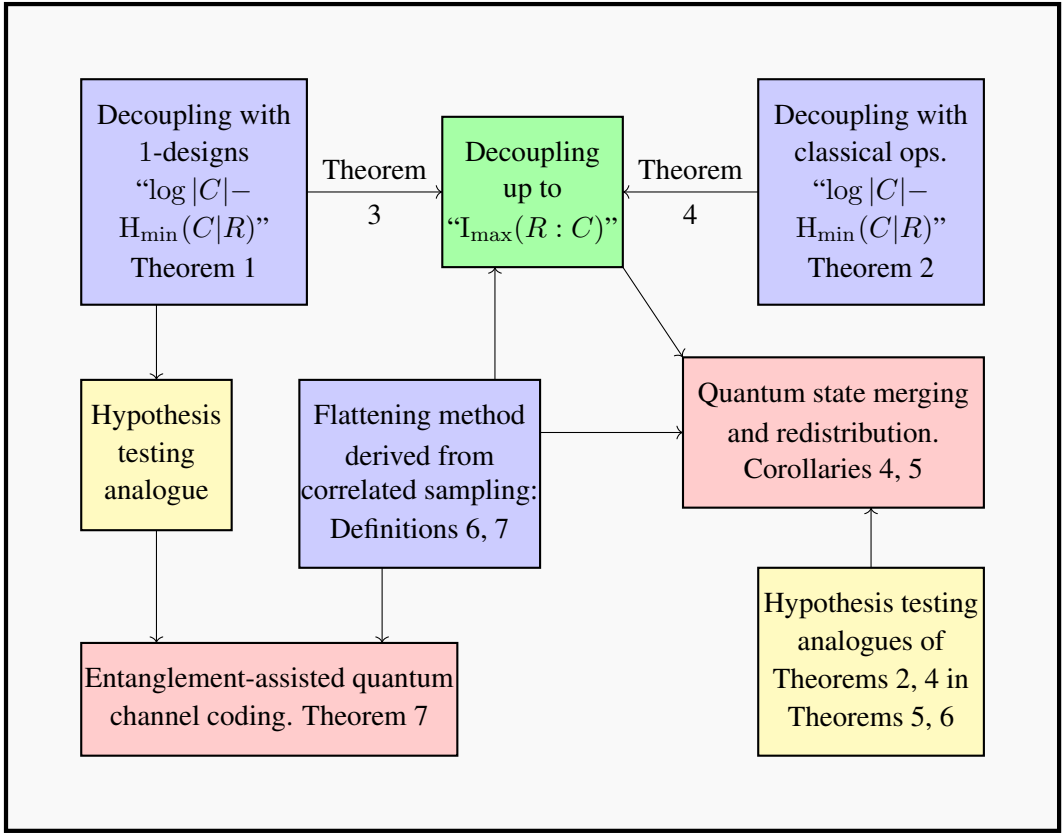

\section{Preliminaries}
\label{sec:prelims}

All the logarithms are evaluated to the base $2$. Consider a finite dimensional Hilbert space $\cH$ endowed with an inner product $\langle \cdot, \cdot \rangle$ (In this paper, we only consider finite dimensional Hilbert-spaces). The $\ell_1$ norm of an operator $X$ on $\cH$ is $\| X\|_1:=\Tr\sqrt{X^{\dagger}X}$ and $\ell_2$ norm is $\| X\|_2:=\sqrt{\Tr XX^{\dagger}}$. A quantum state (or a density matrix or a state) is a positive semi-definite matrix on $\cH$ with trace equal to $1$. It is called {\em pure} if and only if its rank is $1$. A sub-normalized state is a positive semi-definite matrix on $\cH$ with trace less than or equal to $1$. Let $\ket{\psi}$ be a unit vector on $\cH$, that is $\langle \psi,\psi \rangle=1$.  With some abuse of notation, we use $\psi$ to represent the state and also the density matrix $\ketbra{\psi}$, associated with $\ket{\psi}$. Given a quantum state $\rho$ on $\cH$, {\em support of $\rho$}, called $\text{supp}(\rho)$ is the subspace of $\cH$ spanned by all eigenvectors of $\rho$ with non-zero eigenvalues.
 
A {\em quantum register} $A$ is associated with some Hilbert space $\cH_A$. Define $|A| := \dim(\cH_A)$. Let $\mathcal{L}(A)$ represent the set of all linear operators on $\cH_A$. For operators $O, O'\in \cL(A)$, the notation $O\preceq O'$ represents the L\"{o}wner order, that is, $O'-O$ is a positive semi-definite matrix. We denote by $\mathcal{D}(A)$, the set of quantum states on the Hilbert space $\cH_A$. State $\rho$ with subscript $A$ indicates $\rho_A \in \mathcal{D}(A)$. If two registers $A,B$ are associated with the same Hilbert space, we shall represent the relation by $A\equiv B$.  Composition of two registers $A$ and $B$, denoted $AB$, is associated with Hilbert space $\cH_A \otimes \cH_B$.  For two quantum states $\rho\in \mathcal{D}(A)$ and $\sigma\in \mathcal{D}(B)$, $\rho\otimes\sigma \in \mathcal{D}(AB)$ represents the tensor product (Kronecker product) of $\rho$ and $\sigma$. The identity operator on $\cH_A$ (and associated register $A$) is denoted $\id_A$. The maximally mixed state $\frac{\id_A}{|A|}$ on register $A$ is represented by $\mu_A$.

Let $\rho_{AB} \in \mathcal{D}(AB)$. We define
\[ \rho_{B} := \Tr_{A}{\rho_{AB}}
:= \sum_i (\bra{i} \otimes \id_{B})
\rho_{AB} (\ket{i} \otimes \id_{B}) , \]
where $\{\ket{i}\}_i$ is an orthonormal basis for the Hilbert space $\cH_A$.
The state $\rho_B\in \mathcal{D}(B)$ is referred to as the marginal state of $\rho_{AB}$. Unless otherwise stated, a missing register from subscript in a state will represent partial trace over that register. Given a $\rho_A\in\mathcal{D}(A)$, a {\em purification} of $\rho_A$ is a pure state $\rho_{AB}\in \mathcal{D}(AB)$ such that $\Tr_{B}{\rho_{AB}}=\rho_A$. Purification of a quantum state is not unique. Suppose $A\equiv B$. Given $\{\ket{i}_A\}$ and $\{\ket{i}_B\}$ as orthonormal bases over $\cH_A$ and $\cH_B$ respectively, the \textit{canonical purification} of a quantum state $\rho_A$ is $(\rho_A^{\frac{1}{2}}\otimes\id_B)\left(\sum_i\ket{i}_A\ket{i}_B\right)$. 

A quantum {map} $\cE: \mathcal{L}(A)\rightarrow \mathcal{L}(B)$ is a completely positive and trace preserving (CPTP) linear map (mapping states in $\mathcal{D}(A)$ to states in $\mathcal{D}(B)$). A {\em unitary} operator $U_A:\cH_A \rightarrow \cH_A$ is such that $U_A^{\dagger}U_A = U_A U_A^{\dagger} = \id_A$. An {\em isometry}  $V_{A\to B}:\cH_A \rightarrow \cH_B$ is such that $V^{\dagger}V = \id_A$ and $VV^{\dagger} = \id_B$. The set of all unitary operations on register $A$ is  denoted by $\mathcal{U}(A)$. Some standard unitaries are the $\mathsf{X},\mathsf{Z},\mathsf{H}$ (Pauli-$\mathsf{X}$, Pauli-$\mathsf{Z}$ and Hadamard, respectively) gates on qubits, the $\mathsf{CNOT}$ gate on a pair of qubits and the Toffoli gate on three qubits \cite{BarencoBCDMSSSW95}. We will drop the register labels on unitaries unless when it is required. We shall consider the following information theoretic quantities. We consider only normalized states in the definitions below. Let $\eps \in (0,1)$. 

\begin{enumerate}
\item {\bf Fidelity} (\cite{Josza94}, see also \cite{uhlmann76}) For $\rho_A,\sigma_A \in \mathcal{D}(A)$, $$\F(\rho_A,\sigma_A)\defeq\|\sqrt{\rho_A}\sqrt{\sigma_A}\|_1.$$ For classical probability distributions $P = \{p_i\}, Q =\{q_i\}$, $$\F(P,Q)\defeq \sum_i \sqrt{p_i \cdot q_i}.$$
\item {\bf Purified distance} (\cite{GilchristLN05}) For $\rho_A,\sigma_A \in \mathcal{D}(A)$, $$\Pur(\rho_A,\sigma_A) = \sqrt{1-\F^2(\rho_A,\sigma_A)}.$$
\item {\bf $\varepsilon$-ball} For $\rho_A\in \mathcal{D}(A)$, $$\ball{\eps}{\rho_A} \defeq \{\rho'_A\in \mathcal{D}(A)|~\Pur(\rho_A,\rho'_A) \leq \varepsilon\}. $$ 
\item {\bf Smooth max-relative entropy} (\cite{Datta09}, see also \cite{Jain:2009}) For $\rho_A,\sigma_A\in \mathcal{D}(A)$ such that $\text{supp}(\rho_A) \subset \text{supp}(\sigma_A)$, $$ \dmaxeps{\rho_A}{\sigma_A}{\eps}  \defeq  \min_{\rho'_A\in \ball{\eps}{\rho_A}}\min\{ \lambda \in \mathbb{R} : 2^{\lambda} \sigma_A \geq \rho'_A \}.$$  
\item {\bf Hypothesis testing relative entropy} (\cite{BuscemiD10}, see also \cite{HayashiN03})  For $\rho_A,\sigma_A\in \mathcal{D}(A)$, $$\dheps{\rho_A}{\sigma_A}{\eps}  \defeq  \max_{0<\Pi<I, \Tr(\Pi\rho_A)\geq 1-\eps}\log\left(\frac{1}{\Tr(\Pi\sigma_A)}\right).$$  
\item {\bf Max-information} (\cite{CiganovicBR14})  For $\rho_{AB}\in \mathcal{D}(AB)$, $$\imax(A:B)_{\rho} \defeq   \dmax{\rho_{AB}}{\rho_{A}\otimes\rho_{B}} .$$
\item {\bf Smooth max-information} (\cite{CiganovicBR14}) For $\rho_{AB}\in \mathcal{D}(AB)$,  $$\imaxeps{A}{B}{\eps}_{\rho} \defeq \dmax{\rho_{AB}}{\rho_{A}\otimes\rho_{B}}{\eps}  .$$	
\item {\bf Conditional min-entropy} (\cite{Renner05}) For $\rho_{AB}\in \mathcal{D}(AB)$, $$ \hmin{A}{B}_{\rho} \defeq  - \min_{\sigma_B\in \mathcal{D}(B)}\dmax{\rho_{AB}}{\id_{A}\otimes\sigma_{B}} .$$  	
\item {\bf Smooth conditional min-entropy} (\cite{Renner05}) For $\rho_{AB}\in \mathcal{D}(AB)$, $$\hmineps{A}{B}{\eps}_{\rho} \defeq   \max_{\rho^{'} \in \ball{\eps}{\rho}} \hmin{A}{B}_{\rho^{'}} .$$  	
\end{enumerate}

\section{Convex-split with improved resources: basic constructions}
\label{sec:smallrand}

We begin this section by providing a construction of convex-split of a quantum state that uses small amount of additional randomness.

\subsection{Convex-split using a mixture of unitaries from a $1$-design}
\label{subsec:1design}

The unitary $1$-design is defined as follows.
\begin{definition}
\label{def:1design}
Fix a register $C$. A collection of unitaries $\{V_x\}_{x=1}^{|C|^2}$ form a $1$-design if 
$$\frac{1}{|C|^2}\sum_x V_x M V^{\dagger}_x = \Tr(M)\frac{\id_C}{|C|}, \quad \forall M \in \cL(C).$$
\end{definition}

These unitaries have an additional property that they are perfect decouplers, that is,
\begin{equation}
\label{eq:1design}
\frac{1}{|C|^2}\sum_x V_x \rho_{RC} V^{\dagger}_x = \rho_R\otimes\frac{\id_C}{|C|},
\end{equation}
which is evident from Definition \ref{def:1design}. In order to use a small subset of them decoupling, we will require the notion of pairwise independent functions.
\begin{definition}
\label{def:pairwiseunit}
 Let $\{f_j: \cX\times\cX\rightarrow \cX\}_{j=1}^{|\cX|}$ be a family of pairwise independent functions. That is, 
$$\frac{|\{(x_1,x_2):f_j(x_1,x_2) =x , f_k(x_1,x_2)=x'\}|}{|\cX|^2}= \frac{1}{|\cX|^2}, \quad \forall x,x', \quad \forall j\neq k.$$
Introduce registers $X_1\equiv X_2$ such that $|X_1|=|X_2|=|\cX|$.  Let $V^{(j)}:\cH_{CX_1X_2}\rightarrow \cH_{CX_1X_2}$ be defined as $V^{(j)}= \sum_{x_1,x_2}V_{f_j(x_1,x_2)}\otimes\ketbra{x_1, x_2}_{X_1X_2}$. 
\end{definition}
As discussed in \cite[Example 6]{KCN13} or \cite{Lovettnotes}, there exists an efficient construction of pairwise independent function family for any $\cX$ with $|\cX|$ a prime power. In our setting, $|\cX|=|C|^2$. Hence, such a construction exists whenever $\log|C|$ is an integer. 
The following theorem ensures that convex-split can be achieved with small amount of additional resource. Its proof appears in Section \ref{proofs:smallrand}.
\begin{theorem}
\label{paulisplit}
Suppose $\log|C|$ is an integer. Let $\Psi_{RC}$ be a quantum state. Define $k\defeq \dmax{\Psi_{RC}}{\Psi_R\otimes \mu_{C}}$. Define the quantum state
$$\tau_j \defeq V^{(j)}(\Psi_{RC}\otimes \mu_{X_1X_2})V^{(j)\dagger}, \tau\defeq \frac{1}{N}\sum_j \tau_j.$$
It holds that
$$\relent{\tau}{\Psi_R\otimes\mu_C\otimes \mu_{X_1X_2}} \leq \log\left(1+\frac{2^k-1}{N}\right).$$
\end{theorem}

Now we give a canonical example of a unitary $1$-design, which we will use in later sections. 

\begin{definition}
\label{paulis}
Given a register $C$ and a basis $\{\ket{c}\}_{c=0}^{|C|-1}$. Define the Heisenberg-Weyl (HW) unitaries $\{V_{a,b}\}_{a,b=0}^{|Z|-1}$ with $V_{a,b}:\cH_C\rightarrow \cH_C$ as $V_{a,b}\defeq \sum_{c} e^{\frac{2\pi icb}{|C|}}\ket{c+a}_Z\bra{c}_Z$. 
\end{definition}

Following is a well known lemma, which shows that HW unitaries are a $1$-design.

\begin{lemma}
\label{totmix}
For all $\ket{c}_C,\ket{c'}_C,$ it holds that
$$\frac{1}{|C|^2}\sum_{a,b} V_{a,b}\ket{c}\bra{c'}_CV^{\dagger}_{a,b} = \delta_{c,c'}\mu_C.$$
In particular, this implies that for any state $\rho_{RC}$, 
$$\frac{1}{|C|^2}\sum_{a,b} V_{a,b}\rho_{RC}V^{\dagger}_{a,b} = \rho_R\otimes \mu_C.$$ 
\end{lemma}
\begin{proof}
Consider
\begin{eqnarray*}
\frac{1}{|C|^2}\sum_{a,b} V_{a,b}\ket{c}\bra{c'}_CV^{\dagger}_{a,b}&=& \left(\frac{1}{|C|^2}\sum_{a,b} e^{\frac{2\pi i(c-c')b}{|C|}}\ket{c+a}\bra{c'+a}_C\right)\\
&=& \delta_{c,c'}\left(\frac{1}{|C|}\sum_{a}\ket{c+a}\bra{c+a}_C\right)\\
&=& \delta_{c,c'}\mu_Z.
\end{eqnarray*}
In the second equation, we have used Fact \ref{harmonicseries}.
Expand $\rho_{RC} = \sum_{c,c'}\rho^{c,c'}_R\otimes \ket{c}\bra{c'}_C$. Consider
\begin{eqnarray*}
\frac{1}{|C|^2}\sum_{a,b} V_{a,b}\rho_{RC}V^{\dagger}_{a,b}&=& \sum_{c,c'}\rho^{c,c'}_R\otimes \left(\frac{1}{|C|^2}\sum_{a,b} V_{a,b}\ket{c}\bra{c'}_CV^{\dagger}_{a,b}\right)\\
&=& \sum_{c,c'}\delta_{c,c'}\rho^{c,c}_R\otimes \mu_C = \rho_R\otimes \mu_C.
\end{eqnarray*}
Above, $\delta_{z,z'}$ is the delta function. This completes the proof. 
\end{proof}

Above construction uses HW unitaries which also involve a phase. Hence, these unitaries are not classical. Below, we provide a construction that is completely classical, that is, it permutes basis vectors to basis vectors.

\subsection{Convex-split with classical unitaries}
\label{subsec:classicalunit}
Fix a register $C$. Let $\qbit$ be a register with $|\qbit|=2$.  We denote by $\brc$ a register such that $|\brc|\geq |C|^2$ is a prime and $\cH_{\brc}$ is a subspace of $\cH_{\qbit}\otimes\cH_C\otimes\cH_C$. This choice of $\brc$ can be made due to Bertrand's postulate \cite{Chebysev1852}. Let $\{\ket{c}\}_{c=0}^{|C|-1}$ be an arbitrary choice of basis in $\cH_C$, a natural example of which is the computational basis. This ensures that $\{\ket{q}\ket{c}\ket{c'}\}$ with $q\in\{0,1\}$ is a basis on $\cH_{\qbit}\otimes\cH_C\otimes\cH_C$. We construct a basis $\{\ket{i}\}_{i=0}^{|\brc|-1}$ on $\cH_{\brc}$ as follows. We relabel the vector $\ket{0}\ket{c,c'}$ as $\ket{c|C|+c'}$. This gives $|C|^2$ basis vectors for $\brc$. The remaining $|\brc| - |C|^2$ basis vectors are constructed by relabeling $\ket{1}\ket{c,c'}$ as $\ket{|C|^2 + c|C|+c'}$ as long as $|C|^2 + c|C|+c'\leq |\brc|-1$. We note that the constraint $|C|^2 + c|C|+c'\leq |\brc|-1$ is automatically satisfied in our analysis below, as all the additions, subtractions and multiplications appearing below are performed modulo $|\brc|$, unless explicitly stated. Now, introduce registers $C_0, C_1\equiv C$ and $\brc_1, \brc_2 \equiv \brc$, where $\brc_1$ is chosen such that $\cH_{\brc_1}\subset\cH_{\qbit}\otimes \cH_{C_0}\otimes \cH_{C_1}$.    
\begin{definition}
\label{def:Uell}
For an integer $\ell \in \{0,1, \ldots |\brc|-1\}$, define the operation $U_{\ell}: \cH_{\brc_1}\otimes \cH_{\brc_2} \rightarrow \cH_{\brc_1}\otimes \cH_{\brc_2}$ as follows:
$$U_{\ell}\defeq \sum_{i,j} \ket{i + (j-i)\ell}_{\brc_1}\ket{j + (j-i)\ell}_{\brc_2}\bra{i}_{\brc_1}\bra{j}_{\brc_2}.$$
\end{definition}
We choose the convention that the expression in the kets for registers $\brc_1, \brc_2$ are evaluated modulo $|\brc|$. The following lemma shows that the unitaries in Definition \ref{def:Uell} behave in a `cyclic' manner, analogous to the permutations in 
the convex-split lemma from \cite{AnshuDJ14}.
\begin{lemma}
\label{Uprops}
For every $m,\ell\in \{0,1, \ldots |\brc|-1\}$, it holds that $U_\ell$ is a unitary. Furthermore
$$U_{m} U_{\ell} = U_{m+\ell}, \quad U^{\dagger}_{\ell} = U_{-\ell}.$$
\end{lemma}
\begin{proof}
We first show that $U_\ell$ is a unitary. Let $i,i', j, j'$ be such that
$$i + (j-i)\ell = i' + (j'-i')\ell, \quad j + (j-i)\ell = j' + (j'-i')\ell.$$
This can be rearranged to obtain
$$(j-j')\ell +(i-i')(1-\ell) = 0, \quad (j-j')(1+\ell) - (i-i')\ell=0.$$
Multiplying the first equation by $\ell$, the second by $(1-\ell)$ and adding, we obtain $j-j'=0$. Thus, $(i-i')(1-\ell)=0$ and $(i-i')\ell=0$. Adding, we conclude that $i=i'$. Hence, $U_\ell$ is a unitary.

\noindent Now, consider
\begin{eqnarray*}
U_mU_\ell\ket{i}_{\brc_1}\ket{j}_{\brc_2} &=& U_m\ket{i+(j-i)\ell}_{\brc_1}\ket{j+(j-i)\ell}_{\brc_2}\\
&=& \ket{i+(j-i)\ell + (j-i)m}_{\brc_1}\ket{j+(j-i)\ell+(j-i)m}_{\brc_2}\\
&=& \ket{i+(j-i)(\ell+m) }_{\brc_1}\ket{j+(j-i)(\ell+m)}_{\brc_2}\\
&=& U_{m+\ell}\ket{i}_{\brc_1}\ket{j}_{\brc_2}.
\end{eqnarray*}
Thus, $U_mU_\ell = U_{m+\ell}$. Since $U_0=\id$, we conclude $U^{\dagger}_\ell = U_{-\ell}$. This completes the proof.
\end{proof}

Following is an important property of our collection of unitaries and is analogous to Lemma \ref{totmix}.
\begin{lemma}
\label{lem:Uprop2}
For any quantum state $\Psi_{RC_0}$ and any $m\in \{1, \ldots |\brc|-1\}$, 
it holds that $$\Tr_{\brc_2}\left(U_{m}\left(\Psi_{RC_0}\otimes\ketbra{0}_{\qbit}\otimes\mu_{C_1}\otimes \mu_{\brc_2}\right)U^{\dagger}_{m}\right) = \Psi_R\otimes \mu_{\brc_1},$$
where we use the fact that $\cH_{\brc_1}\subseteq \cH_{\qbit}\otimes\cH_{C_0}\otimes\cH_{C_1}$ to change the register label.  
\end{lemma}
\begin{proof}
Define $\delta_{i,i'}\defeq 1$ if $i=i'$ and $0$ otherwise. Consider 
\begin{eqnarray}
\label{crossvanish}
&&\Tr_{\brc_2}\left(U_{m}\left(\ket{i}\bra{i'}_{\brc_1}\otimes \mu_{\brc_2}\right)U^{\dagger}_{m}\right)\nonumber\\
&&= \frac{1}{|\brc|}\sum_{j=0}^{|\brc|-1} \Tr_{\brc_2}\left(U_{m}\left(\ket{i}\bra{i'}_{\brc_1}\otimes \ketbra{j}_{\brc_2}\right)U^{\dagger}_{m}\right)\nonumber\\
&& =\frac{1}{|\brc|}\sum_{j=0}^{|\brc|-1} \Tr_{\brc_2}\bigg(\ket{jm+ i(1-m)}\bra{jm+ i'(1-m)}_{\brc_1}\otimes\ket{j(m+1)- im}\bra{j(m+1)- i'm}_{\brc_2}\bigg)\nonumber\\ 
&& =\frac{1}{|\brc|}\sum_{j=0}^{|\brc|-1} \ket{jm+ i(1-m)}\bra{jm+ i'(1-m)}_{\brc_1}\cdot \delta_{i,i'}\nonumber\\
&&= \frac{1}{|\brc|}\sum_{j=0}^{|\brc|-1} \ketbra{jm+ i(1-m)}_{\brc_1}\cdot \delta_{i,i'}\nonumber\\
&&= \mu_{\brc_1}\cdot \delta_{i,i'},
\end{eqnarray}
where we have used the fact that for $0<m<|\brc|$ and $|\brc|$ prime, the quantity $jm+i(1-m)$ takes all possible values in $\{0,1,\ldots |\brc|-1\}$ as $j$ varies in $\{0,1, \ldots |\brc|-1\}$. For this, observe that for two $j,j'$, $$jm+i(1-m) = j'm+i(1-m)\implies (j-j')m=0,$$ which implies $j=j'$ as $m\neq 0$.
Now, expand $\Psi_{RC_0}= \sum_{c,c'}\Psi^{(c,c')}_R\otimes \ket{c}\bra{c'}_{C_0}$, where $\Psi^{(c,c')}_R$ are some matrices. Observe that $\Psi_R= \Tr_{C_0}(\Psi_{RC_0}) = \sum_{c}\Psi^{(c,c)}_R$. For any $m>0$, using Equation \ref{crossvanish}, we have
\begin{eqnarray*}
&&\Tr_{\brc_2}\left(U_{m}\left(\Psi_{RC_0}\otimes \ketbra{0}_{\qbit}\otimes\mu_{C_1}\otimes \mu_{\brc_2}\right)U^{\dagger}_{m}\right)\\
&& =\sum_{c,c',c_1}\frac{1}{|C|}\Psi^{(c,c')}_R\otimes \Tr_{\brc_2}\left(U_{m}\left(\ket{c}\bra{c'}_{C_0}\otimes\ketbra{0}_{\qbit}\otimes \ketbra{c_1}_{C_1}\otimes \mu_{\brc_2}\right)U^{\dagger}_{m}\right)\\
&&= \sum_{c,c',c_1}\frac{1}{|C|}\Psi^{(c,c')}_R\otimes \mu_{\brc_1}\cdot \delta_{c,c'}\\
&&= \sum_c \Psi^{(c,c)}_R\otimes \mu_{\brc_1} = \Psi_{R}\otimes \mu_{\brc_1},
\end{eqnarray*}
where we have used that fact that $c|C|+c_1 = c'|C|+c_1 \iff c=c'$.
This completes the proof.
\end{proof}
Now, we are in a position to prove our main result. Its proof appears in Section \ref{proofs:smallrand}.
\begin{theorem}
\label{main:theo}
Let $\Psi_{RC}$ be a quantum state and let $k\defeq \dmax{\Psi_{RC}}{\Psi_R\otimes \mu_{C}}$. For a subset $S\subseteq \{0,1,\ldots |\brc|-1\}$ of size $N\defeq |S|$, define the quantum state
$$\tau_{R\brc_1\brc_2}\defeq \frac{1}{N}\sum_{\ell\in S} U_{\ell}\left(\Psi_{RC_0}\otimes\ketbra{0}_{\qbit}\otimes \mu_{C_1}\otimes \mu_{\brc_2}\right) U_{\ell}^{\dagger}.$$
It holds that
$$\relent{\tau_{R\brc_1\brc_2}}{\Psi_R\otimes \mu_{\brc_1}\otimes\mu_{\brc_2}} \leq \log\left(1+\frac{2^{k+1}-1}{N}\right).$$ From Fact \ref{pinsker}, we conclude that
$$\F^2(\tau_{R\brc_1\brc_2},\Psi_R\otimes \mu_{\brc_1}\otimes\mu_{\brc_2}) \geq \frac{1}{1+\frac{2^{k+1}-1}{N}}.$$
\end{theorem}

An immediate corollary is the smooth version of above result. 
\begin{corollary}
\label{cor:convexcomb}
Let $\eps, \delta\in (0,1)$ and $\Psi_{RC}$ be a quantum state. Let $k\defeq \log|C| - \hmineps{C}{R}{\frac{\eps}{2}}_{\Psi} + \log\frac{8}{\eps^3}$ and $N\geq \frac{2^{k+1}}{\delta^2}$. For a set $S\subseteq \{0,1,\ldots |\brc|-1\}$ of size $|S|=N$, define the quantum state
$$\tau_{R\brc_1\brc_2}\defeq \frac{1}{N}\sum_{\ell\in S} U_{\ell}\left(\Psi_{RC_0}\otimes\ketbra{0}_{\qbit}\otimes\mu_{C_1}\otimes \mu_{\brc_2}\right) U_{\ell}^{\dagger}.$$
It holds that
$$\Pur(\tau_{R\brc_1\brc_2}, \Psi_R\otimes \mu_{\brc_1}\otimes \mu_{\brc_2})\leq 2\eps+\delta.$$
\end{corollary}

\subsection{Implementation of the unitary in Definition \ref{def:Uell}}
\label{unitimp}
The circuit size of the decoupling unitary in Definition \ref{def:Uell} can be bounded as follows. The relabeling $\ket{c_0,c_1}_{C_0C_1}\rightarrow \ket{c_0|C|+c_1}_{\brc_1}$ can be performed by the multiplication algorithm of Sch{\"o}nhage and Strassen \cite{SchonS71} using a circuit of size $\mathcal{O}(\log|C|\log\log|C|)$ and depth $\mathcal{O}(\log\log|C|)$.  Thus, we focus on unitary transformation over the basis $\{\ket{i}\}_{i=0}^{|\brc|}$ for $\cH_{\brc_1}$. From Definition \ref{def:Uell}, we have
$$U = \sum_{i,j, \ell} \ket{j\ell+i(1-\ell)}_{\brc_1}\ket{j(\ell+1)-i\ell}_{\brc_2}\bra{i}_{\brc_1}\bra{j}_{\brc_2}\otimes \ketbra{\ell}_L.$$
Using McLaughlin's algorithm \cite{McLaughlin04} based on the algorithm of Sch{\"o}nhage and Strassen \cite{SchonS71} (see \cite[Section 2.4.3]{BrentZ10} for details), and the standard techniques of reversible computing \cite{Toffoli80, SaeediM13}, the following transformation can be achieved with a circuit of size $\mathcal{O}(\log|C|\log\log|C|)$ and depth $\mathcal{O}(\log\log|C|)$:
\begin{eqnarray*}
W_1&:&\ket{i}_{\brc_1}\ket{j}_{\brc_2}\ket{\ell}_L\ket{0}_{\brc'_1}\ket{0}_{\brc'_2}\\&\rightarrow&\ket{i}_{\brc_1}\ket{j}_{\brc_2}
\ket{\ell}_L\ket{0+j\ell+i(1-\ell)}_{\brc'_1}\ket{0+j(\ell+1)-i\ell}_{\brc'_2}\\
&=&\ket{i}_{\brc_1}\ket{j}_{\brc_2}\ket{\ell}_L\ket{j\ell+i(1-\ell)}_{\brc'_1}\ket{j(\ell+1)-i\ell}_{\brc'_2}.
\end{eqnarray*}
$W_1$ also uses $\mathcal{O}(\log|C|)$ ancillary qubits in initial state $\ket{0}$, which are returned in the initial state after the computation. Now, we swap registers $\brc_1,\brc'_1$ and $\brc_2, \brc'_2$:\
\begin{eqnarray*}
S &:&\ket{i}_{\brc_1}\ket{j}_{\brc_2}\ket{\ell}_L\ket{j\ell+i(1-\ell)}_{\brc'_1}\ket{j(\ell+1)-i\ell}_{\brc'_2}\\
&\rightarrow& \ket{j\ell+i(1-\ell)}_{\brc_1}\ket{j(\ell+1)-i\ell}_{\brc_2}\ket{\ell}_L\ket{i}_{\brc'_1}\ket{j}_{\brc'_2}.
\end{eqnarray*}
Swapping two qubits requires three CNOT gates. Hence this operation can be done in depth $3$. Finally we observe that 
$$i= (\ell+1)\cdot\left(j\ell+i(1-\ell)\right) - \ell\cdot\left(j(\ell+1)-i\ell\right)$$ and $$j=\ell\cdot\left(j\ell+i(1-\ell)\right) + (\ell-1)\cdot\left(j(\ell+1)-i\ell\right).$$
Thus, using the aforementioned circuit for modular multiplication and addition, we can achieve the transformation:
\begin{eqnarray*}
W_2&:&\ket{j\ell+i(1-\ell)}_{\brc_1}\ket{j(\ell+1)-i\ell}_{\brc_2}\ket{\ell}_L\ket{i}_{\brc'_1}\ket{j}_{\brc'_2}\\ 
&\rightarrow& \ket{j\ell+i(1-\ell)}_{\brc_1}\ket{j(\ell+1)-i\ell}_{\brc_2}\ket{\ell}_L\ket{i-i}_{\brc'_1}\ket{j-j}_{\brc'_2}\\
&=& \ket{j\ell+i(1-\ell)}_{\brc_1}\ket{j(\ell+1)-i\ell}_{\brc_2}\ket{\ell}_L\ket{0}_{\brc'_1}\ket{0}_{\brc'_2}.
\end{eqnarray*}
Thus, $U$ can be implemented as $U= W_2SW_1$. All of these constructions can be implemented using the Toffoli gate \cite{Toffoli80}.  Hence, the overall size of the circuit is $\mathcal{O}(\log|C|\log\log|C|)$, depth is $\mathcal{O}(\log\log|C|)$ and additional ancillary $\mathcal{O}(\log|C|)$ qubits initialized in $\ket{0}$, that are reset to $\ket{0}$, are used.

\section{Decoupling up to the max-mutual information using a flattening procedure}
\label{sec:maxmutdec}

 We introduce a close variant of the embezzling state \cite{DamH03}.
\begin{definition}
\label{def:embez}
Let $a,n$ be positive integers such that $n\geq a$ and let $D$ be a register satisfying $|D|\geq n-a$. Define
$$\xi^{a:n}_D\defeq \frac{1}{S(a,n)}\sum_{j=a}^n\frac{1}{j}\ketbra{j}_D,$$
where $S(a,n)\defeq \sum_{j=a}^n\frac{1}{j}$ is the normalization factor. Define
$$\ket{\xi^{a:n}}_{D'D}\defeq \frac{1}{\sqrt{S(a:n)}}\sum_{j=a}^n\frac{1}{\sqrt{j}}\ket{j}_{D'}\ket{j}_D$$ as the canonical purification of $\xi^{a:n}_D$, where $D'\equiv D$.
\end{definition}
We have the following claim, which is a variant of the property of embezzling states proved in \cite{DamH03}.
\begin{claim}
\label{embezclose}
Let $\delta\in(0,\frac{1}{15})$ and $a,b,n$ be integers such that $n\geq a^{\frac{1}{\delta}}, a\geq 2$ and $a\geq b$. Fix registers $D,E$ satisfying $|D|\geq n$ and $|E|\geq b$. Let $W_b$ be the unitary that acts as $$W_b\ket{j}_D\ket{0}_E = \ket{\lfloor j/b\rfloor}_D\ket{j\mmod{b}}_E.$$
It holds that 
$$W_b \left(\xi_D^{a:n}\otimes \ketbra{0}_E\right) W_b^{\dagger} \preceq (1+15\delta) \xi^{1:n}_D\otimes \frac{1}{b}\sum_{e=0}^{b-1}\ketbra{e}_E.$$
\end{claim}
\begin{proof}
Consider
\begin{eqnarray*}
&&W_b \left(\xi_D^{a:n}\otimes \ketbra{0}_E\right) W_b^{\dagger}\\
&&=\frac{1}{S(a,n)}\sum_{j=a}^n\frac{1}{j}W_b \left(\ketbra{j}_D\otimes \ketbra{0}_E\right) W_b^{\dagger}\\
&&=\frac{1}{S(a,n)}\sum_{j=a}^n\frac{1}{j}\ketbra{\lfloor j/b\rfloor}_D\otimes \ketbra{j\mmod{b}}_E\\
&&=\frac{1}{S(a,n)}\sum_{j'=\lfloor\frac{a}{b}\rfloor}^{\lfloor\frac{n}{b}\rfloor}\sum_{e=0}^{b-1}\frac{1}{bj'+e}\ketbra{j'}_D\otimes \ketbra{e}_E\\
&&\preceq \frac{1}{S(a,n)}\sum_{j'=\lfloor\frac{a}{b}\rfloor}^{\lfloor\frac{n}{b}\rfloor}\sum_{e=0}^{b-1}\frac{1}{bj'}\ketbra{j'}_D\otimes \ketbra{e}_E\\
&&=\frac{1}{S(a,n)}\sum_{j'=\lfloor\frac{a}{b}\rfloor}^{\lfloor\frac{n}{b}\rfloor}\frac{1}{j'}\ketbra{j'}_D\otimes \sum_{e=0}^{b-1}\frac{1}{b}\ketbra{e}_E \preceq \frac{S(1,n)}{S(a,n)} \xi^{1:n}_D\otimes \sum_{e=0}^{b-1}\frac{1}{b}\ketbra{e}_E.
\end{eqnarray*}
Now, as shown in \cite{Mascheroni1790}, $|S(a,n) - \log\frac{n}{a}| \leq 4$. Thus, 
$$\frac{S(1,n)}{S(a,n)} \leq \frac{\log n +4}{\log n - \log a - 4} \leq \frac{1+4\delta}{1-5\delta} \leq 1+15\delta.$$
This completes the proof.
\end{proof}

Following claim shows how to `unembezzle' a state.
\begin{claim}
\label{unembezclose}
Fix the integers $n,b,a$ as given in Claim \ref{embezclose}. Let the register $D$ satisfy $n^2\geq |D|\geq (n+1)b$. Let $W_b$ be as defined in Claim \ref{embezclose}. It holds that
$$W_b^{\dagger}\left( \xi^{1:n}_D\otimes \frac{1}{b}\sum_{e=0}^{b-1}\ketbra{e}_E\right) W_b \preceq 4\cdot \xi^{1:|D|}_D\otimes\ketbra{0}_E.$$
\end{claim}
\begin{proof}
We observe that $W^{\dagger}_b\ket{j}_D\ket{e}_E = \ket{jb+e}_D\ket{0}_E$ for all $j\leq n$ and $e<b$. We leave the action of $W^{\dagger}_b$ unspecified for $j\geq n, e\geq b$. Consider
\begin{eqnarray*}
W_b^{\dagger}\left( \xi^{1:n}_D\otimes \frac{1}{b}\sum_{e=0}^{b-1}\ketbra{e}_E\right) W_b&=& \frac{1}{S(1,n)}\sum_{j=1}^n\sum_{e=0}^{b-1}\frac{1}{jb}W_b^{\dagger}\ketbra{j}_D\otimes \ketbra{e}_EW_b\\
&=& \frac{1}{S(1,n)}\sum_{j=1}^n\sum_{e=0}^{b-1}\frac{1}{jb}\ketbra{jb+e}_D\otimes\ketbra{0}_E\\
&\preceq & \frac{2}{S(1,n)}\sum_{j=1}^n\sum_{e=0}^{b-1}\frac{1}{jb+e}\ketbra{jb+e}_D\otimes\ketbra{0}_E\\
&\preceq& \frac{2}{S(1,n)}\sum_{j'=1}^{nb+b}\frac{1}{j'}\ketbra{j'}_D\otimes\ketbra{0}_E\\
&\preceq& \frac{2}{S(1,n)}\sum_{j'=1}^{|D|}\frac{1}{j'}\ketbra{j'}_D\otimes\ketbra{0}_E\\
&=& \frac{2S(1, |D|)}{S(1,n)}\xi^{1:|D|}_D\otimes\ketbra{0}_E\\
&\preceq& 4 \xi^{1:|D|}_D\otimes\ketbra{0}_E,
\end{eqnarray*}
where in the last operator inequality, we use the fact that $|D|\leq n^2$. This completes the proof. 
\end{proof}

A `purified version' of above claims is the following restatement of the result in \cite{DamH03}.
\begin{claim}
\label{purembezzle}
Let $\delta\in (0,\frac{1}{25})$. Let $a,b,n$ be positive integers such that $n\geq a^{\frac{1}{\delta}}, a\geq b/\delta$ and let $D$ be a register satisfying $|D|\geq n-a$. Let $\ket{\mu}_{E'E}\defeq \frac{1}{\sqrt{b}}\sum_{e=0}^{b-1}\ket{e}_{E'}\ket{e}_E$. It holds that
$$\Pur\left(\left(W_b\otimes W_b\right)\left(\xi^{a:n}_{D'D}\otimes \ketbra{0}_{E'}\otimes \ketbra{0}_E\right)\left(W^{\dagger}_b\otimes W^{\dagger}_b\right), \xi^{1:n}_{D'D}\otimes \mu_{EE'}\right)\leq 5\sqrt{\delta}.$$
\end{claim}
\begin{proof}
We have
\begin{eqnarray*}
&&\left(W_b\otimes W_b\right)\ket{\xi^{a:n}}_{D'D}\otimes \ket{0}_{E'}\otimes \ket{0}_E\\
&&= \frac{1}{\sqrt{S(a:n)}}\sum_{j=a}^n\frac{1}{\sqrt{j}}\ket{\lfloor j/b\rfloor}_{D'}\ket{\lfloor j/b\rfloor}_D\ket{j\mmod{b}}_{E'}
\ket{j\mmod{b}}_{E}\\
&& = \frac{1}{\sqrt{S(a:n)}}\sum_{j'=\lfloor a/b\rfloor}^{\lfloor n/b\rfloor}\sum_{e=0}^{b-1} \frac{1}{\sqrt{bj'+e}}\ket{j',j'}_{D'D}\ket{e,e}_{E'E}.
\end{eqnarray*}
Since $\ket{\xi^{1:n}}_{D'D}\ket{\mu}_{E'E} = \frac{1}{\sqrt{S(1:n)}}\sum_{j'=1}^n\sum_{e=0}^{b-1}\frac{1}{\sqrt{j'b}}\ket{j',j'}_{D'D}\ket{e,e}_{E'E}$, we have
\begin{eqnarray*}
&&\F\left(\left(W_b\otimes W_b\right)\left(\xi^{a:n}_{D'D}\otimes \ketbra{0}_{E'}\otimes \ketbra{0}_E\right)\left(W^{\dagger}_b\otimes W^{\dagger}_b\right), \xi^{1:n}_{D'D}\otimes \mu\right)\\
&& =\frac{1}{\sqrt{S(a:n)S(1:n)}}\sum_{j'=\lfloor a/b\rfloor}^{\lfloor n/b\rfloor}\sum_{e=0}^{b-1} \frac{1}{j'b\sqrt{1+ \frac{e}{j'b}}}\\
&& \geq \frac{\sqrt{1-\delta}}{\sqrt{S(a:n)S(1:n)}}\sum_{j'=\lfloor a/b\rfloor}^{\lfloor n/b\rfloor}\sum_{e=0}^{b-1} \frac{1}{j'b}\\
&& = \sqrt{1-\delta}\cdot\frac{S(\lfloor a/b\rfloor, \lfloor n/b\rfloor)}{\sqrt{S(a:n)S(1:n)}} \geq \sqrt{1-25\delta},
\end{eqnarray*}
where we use the fact that $|S(a,n) - \log\frac{n}{a}|\leq 4$. This completes the proof.
\end{proof}

We now introduce the following definition, which shows how to extend a suitable quantum state to make it uniform in a subspace.
\begin{definition}
\label{broextend}
{\bf Flattening a quantum state:} Fix a $\gamma\in(0,1)$ such that $\frac{|C|}{\gamma}$ is an integer and a quantum state $\sigma_C\defeq \sum_cq(c)\ketbra{c}_C$ with eigenvalues $q(c)$ that are integer multiples of $\frac{\gamma}{|C|}$. For a register $E$ satisfying $|E|= \frac{|C|}{\gamma}\max_cq(c)$, define the quantum state $\sigma_{CE}$ as follows:
$$\sigma_{CE}\defeq \sum_c q(c) \ketbra{c}_C \otimes \left(\frac{\gamma}{q(c)|C|}\sum_{e=0}^{\frac{q(c)|C|}{\gamma}-1} \ketbra{e}_E\right) = \frac{\gamma}{|C|}\sum_{c}\sum_{e=0}^{\frac{q(c)|C|}{\gamma}-1}\ketbra{c}_C\otimes \ketbra{e}_E.$$
Observe that $\sigma_{CE}$ is uniform in its support.
\end{definition}
 
The flattening of $\sigma_C$ can be realized in a unitary manner as follows. We define some registers and unitaries required for this process.
\begin{definition}
\label{unitaryflat}
Fix $\delta\in(0,\frac{1}{15})$. Let $a\defeq |E|=\frac{|C|}{\gamma}\max_c q(c)$ and $n \defeq a^{\frac{1}{\delta}}$. Introduce a register $D$ satisfying $|D|\geq n$ with the quantum state $\xi^{a:n}_D$ as given in Definition \ref{def:embez}. 
Define the unitary $W: \cH_{CED}\rightarrow \cH_{CED}$ as $$W \defeq \sum_c \ketbra{c}_C\otimes W_{\frac{q(c)|C|}{\gamma}},$$ where $q(c), \gamma$ are given in Definition \ref{broextend} and the unitary $W_{\frac{q(c)|C|}{\gamma}}$ is defined in Claim \ref{embezclose}.
\end{definition}

Flattening is ensured via the following relation, which uses Claim \ref{embezclose}.
\begin{equation}
\label{unitarybros}
W\left(\sigma_C\otimes \ketbra{0}_E\otimes \xi^{a:n}_D\right)W^{\dagger}= \sum_cq(c)\ketbra{c}_{C}\otimes W_{\frac{q(c)|C|}{\gamma}}\left(\ketbra{0}_{E}\otimes \xi^{a:n}_D\right)W^{\dagger}_{\frac{q(c)|C|}{\gamma}} \preceq  (1+15\delta)\sigma_{CE}\otimes \xi^{1:n}_D.
\end{equation}
Given the flattening of a quantum state $\sigma_C$,  Definition \ref{paulis} gives us $|\supp(\sigma_{CE})|^2$ unitaries $V_x:\supp(\sigma_{CE})\rightarrow \supp(\sigma_{CE})$. If $|\supp(\sigma_{CE})|^2= \left(\frac{|C|}{\gamma}\right)^2$ is a prime power, Definition \ref{def:pairwiseunit} gives us the collection of unitaries $V^{\ell}:\supp(\sigma_{CE})\otimes \cH_{X_1X_2}\rightarrow \supp(\sigma_{CE})\otimes\cH_{X_1X_2}$, with $ \log|X_1|=\log|X_2| = \log\left(\frac{|C|}{\gamma}\right)^2$. This allows us to construct the quantum states
\begin{equation}
\label{eq:extendtau}
\tau_\ell\defeq V^{(\ell)}\left(W\left(\Psi_{RC}\otimes \ketbra{0}_{E}\otimes \xi^{a:n}_D\right)W^{\dagger}\otimes \mu_{X_1X_2}\right)V^{(\ell)\dagger},
\end{equation}
Now we prove the following theorem, which is the analogue of Theorem \ref{paulisplit} for a flattened quantum state. Its proof appears in Section \ref{proofs:maxmutdec}.
\begin{theorem}
\label{theo:paulimaxmut}
Fix $\eps\in (0,1), \gamma\in(0,\frac{1}{2})$, $\delta\in(0,\frac{1}{15})$ such that $\frac{|C|}{\gamma}$ is a prime power and quantum states $\Psi_{RC}, \omega_C.$  Let $k\defeq \dmax{\Psi_{RC}}{\Psi_R\otimes\omega_C}$ and $N$ be an integer. Let $\sigma_C$ be the quantum state as constructed in the first part of Fact \ref{nearbygood} using $\omega_C$. For quantum states $\tau_\ell$ as given in Equation \ref{eq:extendtau}, define
$$\tau\defeq \frac{1}{N}\sum_{\ell}\tau_\ell.$$
It holds that
$$\relent{\tau}{\Psi_R\otimes \sigma_{CE}\otimes \xi^{1:n}_D\otimes\mu_{X_1X_2}}\leq 15\delta+\log\left(1+\frac{2^{k+2}-1}{N}\right).$$
Since one can choose $\log|D|=\log n \leq \frac{1}{\delta}\log\frac{|C|}{\gamma}$, the number of qubits of additional registers is $\log|D| + \log|E|+2\log|X_1| \leq  \left(4+\frac{1}{\delta}\right)\log\frac{|C|}{\gamma}$.
\end{theorem}

For later application, we also state a smooth version of Theorem \ref{theo:paulimaxmut}, which is similar to Corollary \ref{cor:convexcomb}.
\begin{corollary}
\label{cor:maxmut}
Fix $\eps\in (0,1), \delta\in (0,\frac{1}{15}), \gamma\in(0,1)$ such that $\frac{|C|}{\gamma}$ is an integer and a quantum state $\Psi_{RC}$. Let $k\defeq \min_{\Psi'_{RC}\in\ball{\eps}{\Psi_{RC}}}\dmax{\Psi'_{RC}}{\Psi_R\otimes \Psi_{C}}$ and $N\defeq \frac{3\cdot 2^{k+2}}{\delta^3}$.  Let $\sigma_C$ be the quantum state as constructed in the first part of Fact \ref{nearbygood} using $\Psi_C$. For quantum states $\tau_\ell$ as given in Equation \ref{eq:extendtau}, define
$$\tau\defeq \frac{1}{N}\sum_{\ell}\tau_\ell.$$ It holds that
$$\Pur(\tau, \Psi_R\otimes \sigma_{CE}\otimes \xi^{1:n}_D\otimes\mu_{X_1X_2})\leq 2\eps+4\sqrt{\delta}.$$
Since one can choose $\log|D|=\log n\leq \frac{1}{\delta}\log\frac{|C|}{\gamma}$, the number of qubits of additional registers is $\log|D| + 2\log|E| \leq  \left(4+\frac{1}{\delta}\right)\log\frac{|C|}{\gamma}$.
\end{corollary}

In a similar manner, we obtain an improved version of Theorem \ref{main:theo}. We first construct the desired states to be used in the statement of the Theorem. For the flattening of a quantum state $\sigma_C$ as given in Definition \ref{broextend}, let  $\cH'_{CE}\defeq \supp(\sigma_{CE})\subset \cH_C\otimes\cH_E$ denote the support of $\sigma_{CE}$. Introduce registers $C_0E_0\equiv CE$ and $C_1E_1\equiv CE$. Let $\qbit$ be a register such that $|\qbit|=2$. Let $\brce$ be a register such that $|\brce|$ is a prime, $\cH_{\brce}\subseteq\cH_{\qbit}\otimes\cH'_{CE}\otimes \cH'_{CE}$ and $\supp(\ketbra{0}_{\qbit})\otimes\cH'_{CE}\otimes \cH'_{CE}\subseteq \cH_{\brce}$. This choice of $\brce$ is guaranteed by Bertrand's postulate \cite{Chebysev1852}. Introduce register $\brce_2, \brce_1\equiv \brce$ such that $\cH'_{C_0E_0}\otimes \cH'_{C_1E_1}\subseteq \cH_{\brce_1}$. We identify the pair $(c,e)$ with an element in $\{0,1,\ldots \frac{|C|}{\gamma}-1\}$ through some one to one mapping and let $\{U_{\ell}\}_{\ell=0}^{|\brce|-1}$ be the unitaries constructed in Definition \ref{def:Uell}, by setting $C\leftarrow \supp(\sigma_{CE})$. Observe that $U_\ell:  \cH_{\brce_1}\otimes\cH_{\brce_2}\rightarrow \cH_{\brce_1}\otimes\cH_{\brce_2}$ are `classical' as long as choice of the preferred basis on $\cH_C$ is the eigenbasis of $\sigma_C$. Define the quantum states
\begin{equation}
\label{eq:tauellprime}
\tau_\ell\defeq U_\ell\left(W\left(\Psi_{RC_0}\otimes \ketbra{0}_{E_0}\otimes \xi^{a:n}_D\right)W^{\dagger}\otimes \ketbra{0}_{\qbit}\otimes \sigma_{C_1E_1}\otimes \mu_{\brce_2}\right)U_\ell^{\dagger}, \quad 
\end{equation}
where $W$ is as given in Definition \ref{unitaryflat} and $|D|\geq n$. We have the following theorem. Its proof appears in Section \ref{proofs:maxmutdec}.
\begin{theorem}
\label{theo:maxmut}
Fix $\eps, \gamma\in(0,1)$, $\delta\in(0,\frac{1}{15})$ such that $\frac{|C|}{\gamma}$ is an integer and quantum states $\Psi_{RC}, \omega_C$ .  Let $k\defeq \dmax{\Psi_{RC}}{\Psi_R\otimes\omega_C}$, $S\subseteq \{0,1,\ldots \frac{|C|^2}{\gamma^2}-1\}$ and $N\defeq |S|$. Let $\sigma_C$ be the quantum state as constructed in the first part of Fact \ref{nearbygood} using $\omega_C$. For the quantum states $\tau_{\ell}$ as constructed in Equation \ref{eq:tauellprime},  define
$$\tau\defeq \frac{1}{N}\sum_{\ell\in S}\tau_\ell.$$
It holds that
$$\relent{\tau}{\Psi_R\otimes \mu_{\brce_1}\otimes \xi^{1:n}_D\otimes \mu_{\brce_2}}\leq 15\delta+\log\left(1+\frac{2^{k+2}-1}{N}\right).$$
Since one can choose $\log|D|=\log n\leq \frac{1}{\delta}\log\frac{|C|}{\gamma}$, the number of qubits of additional registers is $\log|D| + 2\log|\brce| \leq  \left(4+\frac{1}{\delta}\right)\log\frac{|C|}{\gamma}$.
\end{theorem}

We state a smooth version of Theorem \ref{theo:maxmut} which will be used later.

\begin{corollary}
\label{cor:maxmutnew}
Fix $\eps, \gamma\in(0,1)$, $\delta\in(0,\frac{1}{15})$ such that $\frac{|C|}{\gamma}$ is an integer and a quantum state $\Psi_{RC}$ .  Let $k\defeq \min_{\Psi'_{RC}\in\ball{\eps}{\Psi_{RC}}}\dmax{\Psi'_{RC}}{\Psi_R\otimes \Psi_{C}}$ and $N\defeq \frac{3\cdot 2^{k+2}}{\delta^3}$. Let $\sigma_C$ be the quantum state as constructed in the first part of Fact \ref{nearbygood} using $\Psi_C$. For the quantum states $\tau_{\ell}$ as constructed in Equation \ref{eq:tauellprime},  define
$$\tau\defeq \frac{1}{N}\sum_{\ell=1}^N\tau_\ell.$$
It holds that
$$\Pur(\tau,\Psi_R\otimes \mu_{\brce_1}\otimes \xi^{1:n}_D\otimes \mu_{\brce_2})\leq 2\eps + 4\sqrt{\delta}.$$
Since one can choose $\log|D|=\log n\leq \frac{1}{\delta}\log\frac{|C|}{\gamma}$, the number of qubits of additional registers is $\log|D| + 2\log|\brce| \leq  \left(4+\frac{1}{\delta}\right)\log\frac{|C|}{\gamma}$.
\end{corollary}

\section{Analogues of position-based decoding}
\label{sec:posbased}

We now show how to perform hypothesis testing as a dual to Theorem \ref{main:theo}, in analogy with position-based decoding \cite{AnshuJW17CC}. We note that similar construction can achieve a dual to Theorem \ref{paulisplit}, but we do not state it here as it will be constructed in details in Theorem \ref{optchancode}. We have the following theorem. Its proof appears in Section \ref{proofs:posbased}.
\begin{theorem}
\label{theo:posbased}
Let $\eps\in (0,1)$ and $\Psi_{BC}$ be a quantum state. Let $\cS\subseteq \{0,1,\ldots |\brc|-1\}$ such that 
$$|\cS| \leq \frac{\delta^2}{4\eps}2^{\dheps{\Psi_{BC}}{\Psi_B\otimes \mu_C}{\eps}}.$$ For each $\ell\in \cS$, let $\tau_\ell$ be the quantum state defined in Theorem \ref{main:theo} with $\Psi_{RC}\leftarrow \Psi_{BC}$. There exists an POVM $\{\Lambda_{-1}, \Lambda_{\ell}\}_{\ell\in \cS}$ such that
$$\Tr\left(\Lambda_\ell\tau_\ell\right) \geq 1- \eps-4\delta \quad \forall \ell \in \cS.$$
\end{theorem}

Along the lines similar to Theorem \ref{theo:posbased}, we have the following theorem for position-based decoding. We will directly use the registers and unitaries as introduced in Theorem \ref{theo:maxmut}. The proof appears in Section \ref{proofs:posbased}.

\begin{theorem}
\label{cor:posbased}
Let $\eps\in (0,1), \delta\in(0,\frac{1}{15})$ and $\Psi_{BC}, \omega_C$ be quantum states.  Let $\cS\subseteq \{0,1,\ldots |\brc|-1\}$ such that 
$$|\cS| \leq \frac{\delta^2}{4\eps}2^{\dheps{\Psi_{BC}}{\Psi_B\otimes \omega_C}{\eps}}.$$ Let $\sigma_C$ be the quantum state as constructed in the second part of Fact \ref{nearbygood} using $\omega_C$. Let $\tau_\ell$ be the quantum states as defined in Equation \ref{eq:tauellprime}, using the quantum states $\Psi_{RC}\leftarrow\Psi_{BC}$, $\sigma_{CE}$ and by choosing $|D|\leq 2n\cdot|E| \leq 2|E|^{1+\frac{1}{\delta}}$.  There exists a collection of POVM $\{\Lambda_{-1},\Lambda_\ell\}_{\ell\in \cS}$ such that $$\Tr\left(\Lambda_\ell\tau_\ell\right) \geq 1- \eps-64\delta \quad \forall \ell\in \cS.$$
\end{theorem}

\section{Applications}
\label{sec:apps}

\subsection{Entanglement-assisted quantum channel coding}
\label{subsec:chancode}

We show how exponential improvement in entanglement can be obtained for entanglement-assisted quantum channel coding, in comparison to the entanglement required in \cite{AnshuJW17CC}. We begin by defining an entanglement-assisted code. 
\begin{definition}
Fix an $\eps\in(0,1)$ and a positive integer $R$. Let $M'$ be a register of dimension $|M|=2^R$. A $(R, \eps)$ entanglement-assisted code for a quantum channel $\cN_{C\to B}$  consists of a shared entanglement $\ket{\Theta}_{E_AE_B}$ between Alice ($E_A$) and Bob ($E_B$) and
\begin{itemize}
\item An encoding operation $\cE_m: \cL(E_A)\rightarrow \cL(C)$ for each $m\in \{1,2, \ldots 2^R\}$,
\item A decoding operation $\cD: \cL(BE_B)\rightarrow \cL(M')$ which leads to a classical distribution on register $M'$ such that 
$$\Pr\left[M'\neq m\right] \leq \eps, \quad \forall m\in \{1,2, \ldots 2^R\}.$$
\end{itemize}
\end{definition}

We have the following theorem, near-optimality of which is shown by the converse given in \cite{MatthewsW14}. Its proof appears in Section \ref{proofs:apps}.
\begin{theorem}
\label{optchancode}
Let $\eps,\delta'\in(0,1), \delta\in(0,\frac{1}{25}), \gamma\in(0,\frac{1}{2})$. For any pure quantum state $\ket{\Psi}_{AC}$ and
$$R \leq \dheps{\chnl{\Psi_{AC}}}{\chnl{\Psi_A}\otimes\Psi_C}{\eps}-5 -\log\frac{4(\eps+4\gamma^{1/4})}{\delta'},$$
there exists a $(R, \eps + 4\gamma^{1/4}+ \delta'+ 20\sqrt{\delta})$ entanglement-assisted code for a quantum channel $\cN_{A\to B}$. The protocol uses $ \frac{1}{\delta}\log\frac{|A|}{\gamma\cdot\delta}$ qubits of shared entanglement and $4\log|A|$ bits of shared randomness. The latter can be fixed by standard derandomization argument.
\end{theorem}

\subsection{Consequences for quantum state merging and quantum state redistribution}
\label{subsec:stateredist}

Combining Corollary \ref{cor:maxmutnew} (which is a smooth version of Theorem \ref{theo:maxmut}; alternatively we could use Corollary \ref{cor:maxmut}) and Theorem \ref{cor:posbased}, we exponentially improve upon the entanglement cost of the protocol for quantum state redistribution given in \cite{AnshuJW17SR}. Since the proof is similar to that given in \cite{AnshuJW17SR}, we give the statement of the result.
\begin{corollary}
\label{cor:stateredist}
Fix $\eps\in(0,1), \delta \in (0, \frac{1}{15})$ and a pure quantum state $\ket{\Psi}_{RABC}$. There exists an entanglement-assisted one-way protocol in which Alice ($AC$), Bob ($B$) and Reference ($R$) start with the quantum state $\ket{\Psi}_{RABC}$ and Alice communicates a message to Bob such that the final state $\Phi_{RABC}$ between Alice ($A$), Bob ($BC$) and Reference ($R$) satisfies $\Phi_{RABC}\in \ball{4\eps + 65\delta}{\Psi_{RABC}}$. Reference plays no role in the protocol. The number of qubits of shared entanglement required is at most $\left(4+\frac{1}{\delta}\right)\log\frac{|C|}{\delta}$ and  the number of qubits communicated is
$$\min_{\omega_C}\frac{1}{2}\left(\dmaxeps{\Psi'_{RBC}}{\Psi'_{RB}\otimes \omega_{C}}{\eps} -  \dheps{\Psi_{BC}}{\Psi_B\otimes\omega_C}{\eps} + \log\frac{32}{\eps^2\delta^6}\right).$$
\end{corollary}

By the argument in \cite{AnshuDJ14} that shows how a convex-split for $\Psi_{RC}$ can be used to obtain a protocol for the task of quantum state splitting, we obtain the following corollary using Theorem \ref{theo:maxmut}.

\begin{corollary}
\label{cor:statemerge}
Fix $\eps, \delta\in(0,\frac{1}{15})$ and a pure quantum state $\ket{\Psi}_{RAC}$. There exists an entanglement-assisted one-way protocol in which Alice ($AC$) and Reference ($R$) start with the quantum state $\ket{\Psi}_{RAC}$ and Alice communicates a message to Bob such that the final state $\Phi_{RAC}$ between Alice ($A$), Bob ($C$) and Reference ($R$) satisfies $\Phi_{RAC}\in \ball{2\eps + 8\sqrt{\delta}}{\Psi_{RABC}}$. Reference plays no role in the protocol. The number of qubits communicated is
$$\frac{1}{2}\imaxeps{R}{C}{\eps}_{\Psi}+2 + 2\log\frac{1}{\delta}.$$
The number of qubits of entanglement required is at most $\left(4+\frac{1}{\delta}\right)\log\frac{|C|}{\delta}$.
\end{corollary}

Thus, the result improves upon the number of qubits communicated in \cite{Renner11} by an additive factor of $\log\log|C|$ and at the same time achieves the same number of qubits of entanglement required. It achieves the same communication as given in \cite{AnshuDJ14}, but exponentially improves upon the number of qubits of entanglement.

\section{Proofs used in main theorems}

 \subsection{Basic facts used in our proofs}	

We will use the following facts. 
\begin{fact}[Triangle inequality for purified distance~\cite{Tomamichel12}]
\label{fact:trianglepurified}
For states $\rho_A, \sigma_A, \tau_A\in \mathcal{D}(A)$,
$$\Pur(\rho_A,\sigma_A) \leq \Pur(\rho_A,\tau_A)  + \Pur(\tau_A,\sigma_A) . $$ 
\end{fact}

\begin{fact}[Monotonicity under quantum operations~\cite{barnum96,lindblad75}]
	\label{fact:dataprocessing}
For quantum states $\rho$, $\sigma \in \mathcal{D}(A)$, and quantum operation $\cE(\cdot):\mathcal{L}(A)\rightarrow \mathcal{L}(B)$, it holds that
\begin{align*}
	\|\cE(\rho) - \cE(\sigma)\|_1 \leq \|\rho - \sigma\|_1 \quad \mbox{and} \quad \F(\cE(\rho),\cE(\sigma)) \geq \F(\rho,\sigma) \quad \mbox{and} \quad \relent{\rho}{\sigma}\geq \relent{\cE(\rho)}{\cE(\sigma)}.
\end{align*}
\end{fact}

\begin{fact}[Uhlmann's theorem~\cite{uhlmann76}]
\label{uhlmann}
Let $\rho_A,\sigma_A\in \mathcal{D}(A)$. Let $\rho_{AB}\in \mathcal{D}(AB)$ be a purification of $\rho_A$ and $\sigma_{AC}\in\mathcal{D}(AC)$ be a purification of $\sigma_A$. There exists an isometry $V: C \rightarrow B$ such that,
 $$\F(\ketbra{\theta}_{AB}, \ketbra{\rho}_{AB}) = \F(\rho_A,\sigma_A) ,$$
 where $\ket{\theta}_{AB} = (\id_A \otimes V) \ket{\sigma}_{AC}$.
\end{fact}

\begin{fact}[Gentle measurement lemma~\cite{Winter:1999,Ogawa:2002}]
\label{gentlelemma}
Let $\rho$ be a quantum state and $0<A<I$ be an operator. Then 
$$\F(\rho, \frac{A\rho A}{\Tr(A^2\rho)})\geq \sqrt{\Tr(A^2\rho)}.$$
\end{fact}

Following fact implies the Pinsker's inequality.

\begin{fact}[Lemma 5~\cite{Jain:2003a}]
\label{pinsker}
For quantum states $\rho_A,\sigma_A\in\mathcal{D}(A)$, 
$$\F(\rho,\sigma) \geq 2^{-\frac{1}{2}\relent{\rho}{\sigma}}.$$
\end{fact}

\begin{fact}
\label{nearbygood}
Fix a $\gamma\in (0,1)$ and a quantum state $\omega_C$. It holds that
\begin{itemize}
\item  there exists a quantum state $\sigma_C$ such that $\omega_C\preceq \frac{1}{1-\gamma}\sigma_C$ and the eigenvalues of $\sigma_C$ are integer multiples of $\frac{\gamma}{|C|}$. 
\item there exists a quantum state $\sigma_C$ such that $\sigma_C\preceq \frac{1}{1-\gamma}\omega_C$ and the eigenvalues of $\sigma_C$ are integer multiples of $\frac{\gamma}{|C|}$.  
\end{itemize}  
\end{fact}
\begin{proof}
We prove each item as follows. Let $\eta$ be chosen below.
\begin{itemize}
\item Given the quantum state $\omega_C$, we construct an operator $O$ by increasing each eigenvalue of $\omega_C$ to the nearest multiple of $\frac{\eta}{|C|}$, and define $\sigma_C\defeq \frac{O}{\Tr(O)}$. We have $$1=\Tr(\sigma_C)\leq \Tr(O) \leq \Tr(\sigma_C)+|C|\frac{\eta}{|C|} = 1+\eta.$$ Define $\eta'\defeq \Tr(O)-1$ which implies $0\leq \eta' \leq \eta$. The eigenvalues of $\sigma_C$ are integer multiples of $\frac{\eta}{(1+\eta')|C|}$. We choose $\eta$ (which determines $\eta'$ as well) such that $\frac{\eta}{1+\eta'}=\gamma$. This ensures that $\gamma \leq \eta \leq \frac{\gamma}{1-\gamma}$. Furthermore, eigenvalues of $\sigma_C$ are integer multiples of $\frac{\gamma}{|C|}$ and 
$$\omega_C\preceq O = (1+\eta')\sigma_C \preceq (1+\eta)\sigma_C \preceq \frac{1}{1-\gamma}\sigma_C.$$
\item This follows in a similar manner. We construct an operator $O$ by decreasing each eigenvalue of $\omega_C$ to the nearest multiple of $\frac{\eta}{|C|}$, and define $\sigma_C\defeq \frac{O}{\Tr(O)}$. We have $$1= \Tr(\omega_C)\geq \Tr(O) \geq \Tr(\omega_C) - |C|\frac{\eta}{|C|} = 1-\eta.$$ Define $\eta'\defeq 1-\Tr(O)$, which implies $0\leq \eta' \leq \eta$. The eigenvalues of $\sigma_C$ are integer multiples of $\frac{\eta}{(1-\eta')|C|}$. We choose $\eta$ (which determines $\eta'$ as well) such that $\frac{\eta}{1-\eta'}=\gamma$. This ensures that $\frac{\gamma}{1+\gamma} \leq \eta \leq \gamma$. Furthermore, eigenvalues of $\sigma_C$ are integer multiples of $\frac{\gamma}{|C|}$ and 
$$\sigma_C= \frac{1}{1-\eta'}O\preceq \frac{1}{1-\eta'}\sigma_C \preceq \frac{1}{1-\eta}\sigma_C \preceq \frac{1}{1-\gamma}\sigma_C.$$
\end{itemize}
This completes the proof.
\end{proof}

\begin{fact}[\cite{AnshuJW17CC}]
\label{closestatesmeasurement}
Let $\rho_A, \sigma_A \in \cD(\cH_A)$ be quantum states. Let $\Lambda\in \cL(\cH_A)$, $0\preceq \Lambda \preceq \id_A$ be a positive semidefinite operator. Then it holds that
\begin{equation*}
|\sqrt{\tr\left(\Lambda\rho_A\right)}-\sqrt{\tr\left(\Lambda\sigma_A\right)}| \leq \Pur(\rho_A,\sigma_A).
\end{equation*}
\end{fact}

\begin{fact}[\cite{AJMSY16}]
\label{canonicalfid}
Given quantum states $\rho_A, \sigma_A \in \cD(\cH_A)$ and their respective canonical purification $\ket{\rho}_{AB}, \ket{\sigma}_{AB}$ (for $B\equiv A$ and some fixed basis over the registers), 
$$\F(\rho_{AB}, \sigma_{AB})=\Tr\left(\sqrt{\rho_A}\sqrt{\sigma_A}\right)\geq 1-\sqrt{1-\F(\rho_A,\sigma_A)^2} = 1- \Pur(\rho_A, \sigma_A) .$$
\end{fact}

\begin{fact}[\cite{AnshuDJ14}]
\label{elemeq}
Let $\rho_1, \ldots \rho_n, \theta$ be quantum states and $\{p_i\}_i$  be a probability distribution. Define $\rho\defeq \sum_ip_i\rho_i$. Then it holds that
$$\relent{\sum_ip_i\rho_i}{\theta} = \sum_i p_i\left(\relent{\rho_i}{\theta} - \relent{\rho_i}{\rho}\right).$$
\end{fact}

\begin{fact}[Hayashi-Nagaoka inequality~\cite{HayashiN03}]
\label{haynag}
Fix a $c>1$ and an integer $N>0$. Let $\{\Omega_0,\ldots \Omega_{N-1}\}_{i=0}$ be a collection of positive semi-definite operators. Define $$\Lambda_i\defeq \left(\sum_{i'}\Omega_{i'}\right)^{-\frac{1}{2}}\Omega_i\left(\sum_{i'}\Omega_{i'}\right)^{-\frac{1}{2}}$$
and $\Lambda_{-1}$ be the projector orthogonal to the support of $\sum_{i'}\Omega_{i'}$. The operators $\{\Lambda_{-1},\Lambda_0,\ldots\Lambda_{N-1}\}$ form a POVM. Then 
$$\id - \Lambda_{i}\preceq (1+c)(\id-\Omega_i) + (1+c+c^{-1})\left(\sum_{i'\neq i}\Omega_{i'}\right).$$
\end{fact}

\begin{fact}[Transpose method]
\label{transposetrick}
Let $C, C'$ be registers such that $C\equiv C'$. Let $\ket{\Phi}_{CC'}$ be the maximally entangled state on $\cH_C\otimes\cH_{C'}$ with $\Phi_C=\mu_C$ and $\Phi_{C'}=\mu_{C'}$. For any unitary $U:\cH_C\rightarrow \cH_C$, there exists a unitary $U^T: \cH_{C'}\rightarrow \cH_{C'}$ such that 
$$(U\otimes \id_{C'})\ket{\Phi}_{CC'} = (\id_C\otimes U^T)\ket{\Phi}_{CC'}.$$
\end{fact}

The following fact was stated in \cite[Claim 4]{AnshuJW17CC}, with proof adapted from \cite{CiganovicBR14}.

\begin{fact}
\label{nearbygooddmax}
Let $\delta\in (0,1)$. For quantum states $\sigma_A, \sigma_B, \rho_{AB}$, there exists a quantum state $\rho'_{AB}\in\ball{\delta}{\rho_{AB}}$ such that
$$\dmax{\rho'_{AB}}{\rho'_A\otimes \sigma_B} \leq \dmax{\rho_{AB}}{\sigma_A\otimes \sigma_B}+ \log\frac{3}{\delta^2}.$$
\end{fact}

\begin{fact}
\label{dmaxhmin}
Let $\eps\in(0,1)$ and $\rho_{AB}$ be a quantum state. It holds that
$$\min_{\rho'\in\ball{2\eps}{\rho}}\dmax{\rho'_{AB}}{\rho'_A\otimes \frac{\id_B}{|B|}} \leq \log|B| - \hmineps{B}{A}{\eps}_{\rho} + 3\log\frac{2}{\eps}.$$
\end{fact}
\begin{proof}
From \cite[Fact 12]{AnshuJW17MC} (a corollary of an argument in \cite{CiganovicBR14}), for every quantum state $\sigma_A$, it holds that
$$\min_{\rho'\in\ball{2\eps}{\rho}}\dmax{\rho'_{AB}}{\rho'_A\otimes \frac{\id_B}{|B|}} \leq \dmaxeps{\rho_{AB}}{\sigma_A\otimes \frac{\id_B}{|B|}}{\eps} + 3\log\frac{2}{\eps}.$$ Minimizing over all $\sigma_A$, we have
$$\min_{\rho'\in\ball{2\eps}{\rho}}\dmax{\rho'_{AB}}{\rho'_A\otimes \frac{\id_B}{|B|}} \leq \log|B|+  \min_{\sigma_A}\dmaxeps{\rho_{AB}}{\sigma_A\otimes \id_B}{\eps} + 3\log\frac{2}{\eps}.$$
The proof now concludes by the definition of $ \hmineps{B}{A}{\eps}_{\rho}$.
\end{proof}

\begin{fact}
\label{harmonicseries}
It holds that for $0<a< C$, $$\sum_{b=0}^{C-1}e^{\frac{2\pi iab}{C}}=0.$$
\end{fact}
\begin{proof}
Let $S\defeq \sum_{b=0}^{C-1}e^{\frac{2\pi iab}{C}}$. We have $$e^{\frac{2\pi ia}{C}}S= \sum_{b=0}^{C-1}e^{\frac{2\pi ia(b+1)}{C}}= \sum_{b=1}^{C-1}e^{\frac{2\pi iab}{C}} + e^{\frac{2\pi iaC}{C}}=1+\sum_{b=1}^{C-1}e^{\frac{2\pi iab}{C}}=S.$$
Thus, $(1-e^{\frac{2\pi ia}{C}})S=0$. Since $e^{\frac{2\pi ia}{C}}\neq 1$ for $0<a<C$, the proof concludes.
\end{proof}

\subsection{Proofs in Section \ref{sec:smallrand}}
\label{proofs:smallrand}

\begin{proof}[\bf Proof of Theorem \ref{paulisplit}]
Observe that $V^{(j)}$ acts controlled on registers $X_1,X_2$. Thus, 
\begin{eqnarray*}
\relent{\tau}{\Psi_R\otimes\mu_C\otimes \mu_{X_1X_2}} = \frac{1}{|\cX|^2}\sum_{x_1,x_2}\relent{\frac{1}{N}\sum_jV_{f_j(x_1,x_2)}\Psi_{RC}V^{\dagger}_{f_j(x_1,x_2)}}{\Psi_{R}\otimes\mu_C}.
\end{eqnarray*}
Using Fact \ref{elemeq}, we have
\begin{eqnarray*}
&&\relent{\frac{1}{N}\sum_jV_{f_j(x_1,x_2)}\Psi_{RC}V^{\dagger}_{f_j(x_1,x_2)}}{\Psi_{R}\otimes\mu_C}\\
&&= \frac{1}{N}\sum_j\bigg(\relent{V_{f_j(x_1,x_2)}\Psi_{RC}V^{\dagger}_{f_j(x_1,x_2)}}{\Psi_{R}\otimes\mu_C} \\
&& - \relent{V_{f_j(x_1,x_2)}\Psi_{RC}V^{\dagger}_{f_j(x_1,x_2)}}{\frac{1}{N}\sum_kV_{f_k(x_1,x_2)}\Psi_{RC}V^{\dagger}_{f_k(x_1,x_2)}}\bigg)\\
&&= \frac{1}{N}\sum_j\bigg(\relent{\Psi_{RC}}{\Psi_{R}\otimes V^{\dagger}_{f_j(x_1,x_2)}\mu_CV_{f_j(x_1,x_2)}} 
\\ &&- \relent{\Psi_{RC}}{\frac{1}{N}\Psi_{RC}+\sum_{k\neq j}\frac{1}{N}V^{\dagger}_{f_j(x_1,x_2)}V_{f_k(x_1,x_2)}\Psi_{RC}V^{\dagger}_{f_k(x_1,x_2)}V_{f_j(x_1,x_2)}}\bigg)\\
&&=\frac{1}{N}\sum_j\left(\relent{\Psi_{RC}}{\Psi_{R}\otimes\mu_C} - \relent{\Psi_{RC}}{\frac{1}{N}\Psi_{RC}+\sum_{k\neq j}\frac{1}{N}V^{\dagger}_{f_j(x_1,x_2)}V_{f_k(x_1,x_2)}\Psi_{RC}V^{\dagger}_{f_k(x_1,x_2)}V_{f_j(x_1,x_2)}}\right)\\
&&=\relent{\Psi_{RC}}{\Psi_{R}\otimes\mu_C} - \frac{1}{N}\sum_j\left(\relent{\Psi_{RC}}{\frac{1}{N}\Psi_{RC}+\frac{1}{N}\sum_{k\neq j}V^{\dagger}_{f_j(x_1,x_2)}V_{f_k(x_1,x_2)}\Psi_{RC}V^{\dagger}_{f_k(x_1,x_2)}V_{f_j(x_1,x_2)}}\right).\\
\end{eqnarray*}
Thus,
\begin{eqnarray*}
&&\relent{\tau}{\Psi_R\otimes\mu_C\otimes \mu_{X_1X_2}}=\relent{\Psi_{RC}}{\Psi_{R}\otimes\mu_C}\\
&& - \frac{1}{N|\cX|^2}\sum_j\sum_{x_1,x_2}\left(\relent{\Psi_{RC}}{\frac{1}{N}\Psi_{RC}+\frac{1}{N}\sum_{k\neq j}V^{\dagger}_{f_j(x_1,x_2)}V_{f_k(x_1,x_2)}\Psi_{RC}V^{\dagger}_{f_k(x_1,x_2)}V_{f_j(x_1,x_2)}}\right)\\
&&\leq \relent{\Psi_{RC}}{\Psi_{R}\otimes\mu_C}\\
&& - \frac{1}{N}\sum_j\left(\relent{\Psi_{RC}}{\frac{1}{N}\Psi_{RC}+\frac{1}{N|\cX|^2}\sum_{k\neq j}\sum_{x_1, x_2}V^{\dagger}_{f_j(x_1,x_2)}V_{f_k(x_1,x_2)}\Psi_{RC}V^{\dagger}_{f_k(x_1,x_2)}V_{f_j(x_1,x_2)}}\right),
\end{eqnarray*}
where we have used the convexity of relative entropy. From the pairwise independent property of the family of functions, this simplifies to
\begin{eqnarray}
\label{eq:changexxprime}
&&\relent{\tau}{\Psi_R\otimes\mu_C\otimes \mu_{X_1X_2}}\leq \relent{\Psi_{RC}}{\Psi_{R}\otimes\mu_C}\nonumber\\
&& - \frac{1}{N}\sum_j\left(\relent{\Psi_{RC}}{\frac{1}{N}\Psi_{RC}+\frac{1}{N}\sum_{k\neq j}\sum_{x,x'}\frac{|\{(x_1,x_2):f_j(x_1,x_2) =x , f_k(x_1,x_2)=x'\}|}{|\cX|^2}V^{\dagger}_{x}V_{x'}\Psi_{RC}V^{\dagger}_{x'}V_{x}}\right)\nonumber\\
&& = \relent{\Psi_{RC}}{\Psi_{R}\otimes\mu_C}- \frac{1}{N}\sum_j\left(\relent{\Psi_{RC}}{\frac{1}{N}\Psi_{RC}+\frac{1}{N|\cX|^2}\sum_{k\neq j}\sum_{x,x'}V^{\dagger}_{x}V_{x'}\Psi_{RC}V^{\dagger}_{x'}V_{x}}\right).
\end{eqnarray}
 Equation \ref{eq:1design} ensures that
\begin{eqnarray*}
&&\relent{\tau}{\Psi_R\otimes\mu_C\otimes \mu_{X_1X_2}}\leq \relent{\Psi_{RC}}{\Psi_{R}\otimes\mu_C}\\
&&-\frac{1}{N}\sum_j\left(\relent{\Psi_{RC}}{\frac{1}{N}\Psi_{RC}+\frac{1}{N|\cX|^2}\sum_{k\neq j}\sum_{x,x'}V^{\dagger}_{x}V_{x'}\Psi_{RC}V^{\dagger}_{x'}V_{x}}\right)\\
&&= \relent{\Psi_{RC}}{\Psi_{R}\otimes\mu_C} - \frac{1}{N}\sum_j\left(\relent{\Psi_{RC}}{\frac{1}{N}\Psi_{RC}+\frac{1}{N}\sum_{k\neq j}\Psi_R\otimes\mu_C}\right)\\
&&= \relent{\Psi_{RC}}{\Psi_{R}\otimes\mu_C} - \relent{\Psi_{RC}}{\frac{1}{N}\Psi_{RC}+\frac{N-1}{N}\Psi_R\otimes\mu_C}.
\end{eqnarray*}
Using the inequality $\Psi_{RC}\preceq 2^k\Psi_R\otimes\mu_C$ and the operator monotonicity of logarithm \cite{carlen}, we conclude that
\begin{eqnarray*}
&&\relent{\tau}{\Psi_R\otimes\mu_C\otimes \mu_{X_1X_2}}\\
&&\leq\relent{\Psi_{RC}}{\Psi_{R}\otimes\mu_C} - \relent{\Psi_{RC}}{\Psi_R\otimes\mu_C} + \log\left(1+\frac{2^k-1}{N}\right)\\
&& = \log\left(1+\frac{2^k-1}{N}\right).
\end{eqnarray*}
This completes the proof.
\end{proof}

\begin{proof}[\bf Proof of Theorem \ref{main:theo}]
By definition of $k$, we have $\Psi_{RC_0}\preceq 2^k\Psi_R\otimes \mu_{C_0}$. This implies 
\begin{equation}
\label{biggerdmax}
\Psi_{RC_0}\otimes \ketbra{0}_{\qbit}\otimes\mu_{C_1}\preceq 2^k\Psi_R\otimes\mu_{C_0}\otimes\ketbra{0}_{\qbit}\otimes\mu_{C_1} \preceq \frac{|\brc_1|}{|C_0||C_1|}2^{k}\Psi_R\otimes \mu_{\brc_1}\preceq 2^{k+1}\Psi_R\otimes \mu_{\brc_1}.
\end{equation}
Using Fact \ref{elemeq}, we have 
\begin{eqnarray}
\label{convupbound}
&&\relent{\tau_{R\brc_1\brc_2}}{\Psi_R\otimes \mu_{\brc_1}\otimes\mu_{\brc_2}}\nonumber\\
 &&= \frac{1}{N}\sum_{\ell\in S}\bigg(\relent{U_{\ell}\left(\Psi_{RC_0}\otimes \ketbra{0}_{\qbit}\otimes\mu_{C_1}\otimes \mu_{\brc_2}\right) U_{\ell}^{\dagger}}{\Psi_R\otimes \mu_{\brc_1}\otimes\mu_{\brc_2}}\nonumber\\ 
&& \hspace{3cm}-\relent{U_{\ell}\left(\Psi_{RC_0}\otimes\ketbra{0}_{\qbit}\otimes \mu_{C_1}\otimes\mu_{\brc_2}\right) U_{\ell}^{\dagger}}{\tau_{R\brc_1\brc_2}}\bigg)\nonumber\\
&&= \frac{1}{N}\sum_{\ell\in S}\bigg(\relent{\Psi_{RC_0}\otimes\ketbra{0}_{\qbit}\otimes\mu_{C_1}\otimes \mu_{\brc_2} }{\Psi_R\otimes U^{\dagger}_{\ell}\left(\mu_{\brc_1}\otimes\mu_{\brc_2}\right)U_{\ell}}\nonumber\\
&&\hspace{3cm}-\relent{\Psi_{RC_0}\otimes\ketbra{0}_{\qbit}\otimes\mu_{C_1}\otimes \mu_{\brc_2}}{U^{\dagger}_{\ell}\tau_{R\brc_1\brc_2}U_{\ell}}\bigg)\nonumber\\
&&\leq \frac{1}{N}\sum_{\ell\in S}\bigg(\relent{\Psi_{RC_0}\otimes\ketbra{0}_{\qbit}\otimes\mu_{C_1}\otimes\mu_{\brc_2} }{\Psi_R\otimes U^{\dagger}_{\ell}\left(\mu_{\brc_1}\otimes\mu_{\brc_2}\right)U_{\ell}}\nonumber\\
&&\hspace{3cm}-\relent{\Psi_{RC_0}\otimes\ketbra{0}_{\qbit}\otimes\mu_{C_1}}{\Tr_{\brc_2}\left(U^{\dagger}_{\ell}\tau_{R\brc_1\brc_2}U_{\ell}\right)}\bigg).\nonumber\\
\end{eqnarray}
Since $\mu_{\brc_1}\otimes\mu_{\brc_2}$ is maximally mixed in the support of $U_\ell$, 
\begin{equation}
\label{maxmixinv}
 U^{\dagger}_{\ell}\left(\mu_{\brc_1}\otimes\mu_{\brc_2}\right)U_{\ell} = \mu_{\brc_1}\otimes \mu_{\brc_2}.
\end{equation}
Moreover, from Lemma \ref{Uprops} we have
\begin{eqnarray*}
&&\Tr_{\brc_2}\left(U^{\dagger}_{\ell}\tau_{R\brc_1\brc_2}U_{\ell}\right) \\
&&= \frac{1}{N}\sum_{m\in S}\Tr_{\brc_2}\left(U^{\dagger}_{\ell}U_m\left(\Psi_{RC_0}\otimes\ketbra{0}_{\qbit}\otimes\mu_{C_1}\otimes \mu_{\brc_2}\right)U^{\dagger}_mU_{\ell}\right)\\
&&= \frac{1}{N}\Psi_{RC_0}\otimes\ketbra{0}_{\qbit}\otimes\mu_{C_1} + \frac{1}{N}\sum_{m\in S, m\neq\ell}\Tr_{\brc_2}\left(U_{m-\ell}\left(\Psi_{RC_0}\otimes\ketbra{0}_{\qbit}\otimes\mu_{C_1}\otimes \mu_{\brc_2}\right)U^{\dagger}_{m-\ell}\right)\\
&&\preceq \frac{2^{k+1}}{N}\Psi_R\otimes\mu_{\brc_1}+\frac{1}{N}\sum_{m\in S, m\neq \ell}\Tr_{\brc_2}\left(U_{m-\ell}\left(\Psi_{RC_0}\otimes\ketbra{0}_{\qbit}\otimes\mu_{C_1}\otimes \mu_{\brc_2}\right)U^{\dagger}_{m-\ell}\right),
\end{eqnarray*}
where in last operator inequality, we have used Equation \ref{biggerdmax}. Using Lemma \ref{lem:Uprop2}, we conclude that
$$\Tr_{\brc_2}\left(U^{\dagger}_{\ell}\tau_{R\brc_1\brc_2}U_{\ell}\right) \preceq \frac{2^{k+1}}{N}\Psi_R\otimes\mu_{\brc_1} + \frac{N-1}{N}\Psi_{R}\otimes \mu_{\brc_1} = \left(1+\frac{2^{k+1}-1}{N}\right)\Psi_{R}\otimes \mu_{\brc_1}.$$
Since logarithm is operator monotone \cite{carlen},
\begin{eqnarray*}
&&\relent{\Psi_{RC_0}\otimes\ketbra{0}_{\qbit}\otimes\mu_{C_1}}{\Tr_{\brc_2}\left(U^{\dagger}_{\ell}\tau_{R\brc_1\brc_2}U_{\ell}\right)} \\
&&\geq \log\left(1+\frac{2^k-1}{N}\right) + \relent{\Psi_{RC_0}\otimes\ketbra{0}_{\qbit}\otimes\mu_{C_1}}{\Psi_R\otimes \mu_{\brc_1}}.
\end{eqnarray*}
 From Equations \ref{convupbound} and \ref{maxmixinv},
\begin{eqnarray*}
&&\relent{\tau_{R\brc_1\brc_2}}{\Psi_R\otimes \mu_{\brc_1}\otimes\mu_{\brc_2}}\\ 
&& \leq \frac{1}{N}\sum_{\ell\in S}\bigg(\relent{\Psi_{RC_0}\otimes \ketbra{0}_{\qbit}\otimes\mu_{C_1} \otimes\mu_{\brc_2}}{\Psi_R\otimes \mu_{\brc_1}\otimes\mu_{\brc_2}}- \relent{\Psi_{RC_0}\otimes\ketbra{0}_{\qbit}\otimes\mu_{C_1}}{\Psi_R\otimes \mu_{\brc_1}}\bigg)\\
&& \hspace{4cm}+ \log\left(1+\frac{2^{k+1}-1}{N}\right)\\
&& = \frac{1}{N}\sum_{\ell\in S}\bigg(\relent{\Psi_{RC_0}\otimes\ketbra{0}_{\qbit}\otimes\mu_{C_1}}{\Psi_R\otimes \mu_{\brc_1}}- \relent{\Psi_{RC_0}\otimes\ketbra{0}_{\qbit}\otimes\mu_{C_1}}{\Psi_R\otimes \mu_{\brc_1}}\bigg) \\
&& \hspace{4cm}+ \log\left(1+\frac{2^{k+1}-1}{N}\right)\\
&&=\log\left(1+\frac{2^{k+1}-1}{N}\right).
\end{eqnarray*}
This completes the proof.
\end{proof}

\begin{proof}[\bf Proof of Corollary \ref{cor:convexcomb}]
From Fact \ref{dmaxhmin}, we conclude that
\begin{eqnarray*}
&&\min_{\Psi'_{RC}\in\ball{\eps}{\Psi_{RC}}}\dmax{\Psi'_{RC}}{\Psi'_R\otimes \mu_{C}}\leq \log|C| - \hmineps{C}{R}{\frac{\eps}{2}}_{\Psi} + \log\frac{8}{\eps^3}=k'.
\end{eqnarray*} 
Let $\Psi'_{RC}$ be the quantum state achieving the infimum above.  Define
$$\tau'_{R\brc_1\brc_2}\defeq \frac{1}{N}\sum_{\ell\in S} U_{\ell}\left(\Psi'_{RC_0}\otimes\ketbra{0}_{\qbit}\otimes\mu_{C_1}\otimes \mu_{\brc_2}\right) U_{\ell}^{\dagger}.$$
We use Theorem \ref{main:theo} to conclude that
$$\Pur(\tau'_{R\brc_1\brc_2}, \Psi'_R\otimes \mu_{\brc_1}\otimes \mu_{\brc_2})\leq \delta.$$
By triangle inequality for purified distance, this implies that 
$$\Pur(\tau_{R\brc_1\brc_2}, \Psi_R\otimes \mu_{\brc_1}\otimes \mu_{\brc_2})\leq 2\eps+\delta.$$
This concludes the proof.
\end{proof}

\subsection{Proofs in Section \ref{sec:maxmutdec}}
\label{proofs:maxmutdec}

\begin{proof}[\bf Proof of Theorem \ref{theo:paulimaxmut}]
From Fact \ref{nearbygood}, we have that the eigenvalues of $\sigma_C$ are integer multiples of $\frac{\gamma}{|C|}$ and $$\omega_C\preceq \frac{1}{1-\gamma}\sigma_C\implies \Psi_{RC}\preceq \frac{1}{1-\gamma}2^k\Psi_R\otimes \sigma_C\preceq 2^{k+1}\Psi_R\otimes \sigma_C.$$ Consider, 
\begin{eqnarray*}
W\left(\Psi_{RC}\otimes \ketbra{0}_{E}\otimes \xi^{a:n}_D\right)W^{\dagger}&\preceq& 2^{k+1}\Psi_R\otimes W\left(\sigma_C\otimes \ketbra{0}_{E}\otimes \xi^{a:n}_D\right)W^{\dagger}\\
&\overset{(a)}\preceq& 2^{k+1}(1+15\delta)\Psi_R\otimes \sigma_{CE}\otimes \xi^{1:n}_D\\
&\preceq& 2^{k+2}\Psi_R\otimes \sigma_{CE}\otimes \xi^{1:n}_D,
\end{eqnarray*}
where $(a)$ uses Equation \ref{unitarybros}. Expand $\Psi_{RC}= \sum_{c,c'}\Psi^{(c,c')}_R\otimes \ket{c}\bra{c'}_C$. For convenience, set $b(c)\defeq q(c)|C|/\gamma$. Consider
\begin{eqnarray}
\label{embezcrossvanish}
&&\frac{1}{|X|}\sum_x V_xW\left(\Psi_{RC}\otimes \ketbra{0}_{E}\otimes \xi^{a:n}_D\right)W^{\dagger}V^{\dagger}_x\nonumber\\
&&=\frac{1}{S(a,n)}\sum_{j=a}^n\frac{1}{j}\sum_{c,c'}\Psi^{(c,c')}_R\otimes \frac{1}{|X|}\sum_xV_x\bigg(\ket{c}\bra{c'}_C\otimes \ket{j \mmod{b(c)}}\bra{j\mmod{b(c')}}_D\nonumber\\ 
&& \hspace{7cm}\otimes\ket{\lfloor j/b(c)\rfloor}\bra{\lfloor j/b(c')\rfloor}_E\bigg)V^{\dagger}_x\nonumber\\
&&\overset{(a)}=\sum_{c,c'}\Psi^{(c,c)}_R\otimes \delta_{c,c'}\sigma_{CE}\otimes\frac{1}{S(a,n)}\sum_{j=a}^n\frac{1}{j}\ket{\lfloor j/b(c)\rfloor}\bra{\lfloor j/b(c)\rfloor}_D\nonumber\\
&&= \sum_{c}\Psi^{(c,c)}_{R}\otimes \sigma_{CE}\otimes\frac{1}{S(a,n)}\sum_{j=a}^n\frac{1}{j}\ket{\lfloor j/b(c)\rfloor}\bra{\lfloor j/b(c)\rfloor}_D\nonumber\\
&&\overset{(b)}\preceq (1+15\delta)\sum_{c}\Psi^{(c,c)}_{R}\otimes \sigma_{CE}\otimes\xi^{1:n}_D\nonumber\\
&&= (1+15\delta)\Psi_{R}\otimes \sigma_{CE}\otimes \xi^{1:n}_D.
\end{eqnarray}
The equality $(a)$ uses Lemma \ref{totmix}. The operator inequality $(b)$ uses the fact that $\frac{S(1,n)}{S(a,n)}\leq (1+15\delta)$, as given in Claim \ref{embezclose}.  The rest of the argument is identical to Theorem \ref{paulisplit}, up to the factor of $(1+15\delta)$ induced by above operator inequality.   This completes the proof.
\end{proof}

\begin{proof}[\bf Proof of Theorem \ref{theo:maxmut}]
Fact \ref{nearbygood} ensures that  $$\omega_C\preceq  \frac{1}{1-\gamma}\sigma_C\implies \Psi_{RC}\preceq \frac{1}{1-\gamma}2^k\Psi_R\otimes \sigma_C\preceq 2^{k+1}\Psi_R\otimes \sigma_C.$$
Using Claim \ref{embezprop2} and Fact \ref{elemeq}, we proceed similar to Theorem \ref{main:theo}.
\begin{eqnarray}
\label{elemeqagain}
&&\relent{\tau}{\Psi_R\otimes \mu_{\brce_1}\otimes \xi^{1:n}_D\otimes \mu_{\brce_2}}\nonumber\\
&&=\frac{1}{N}\sum_{\ell\in S}\left(\relent{\tau_\ell}{\Psi_R\otimes \mu_{\brce_1}\otimes \xi^{1:n}_D\otimes \mu_{\brce_2}} - \relent{\tau_\ell}{\tau}\right)\nonumber\\
&&=\frac{1}{N}\sum_{\ell\in S}\bigg(\relent{W\left(\Psi_{RC_0}\otimes \ketbra{0}_{E_0}\otimes \xi^{a:n}_D\right)W^{\dagger}\otimes\ketbra{0}_{\qbit}\otimes \sigma_{C_1E_1}\otimes \mu_{\brce_2}}{\Psi_R\otimes \mu_{\brce_1}\otimes \xi^{1:n}_D\otimes \mu_{\brce_2}}\nonumber\\
&& \hspace{2cm} - \relent{W\left(\Psi_{RC_0}\otimes \ketbra{0}_{E_0}\otimes \xi^{a:n}_D\right)W^{\dagger}\otimes \ketbra{0}_{\qbit}\otimes\sigma_{C_1E_1}\otimes \mu_{\brce_2}}{U^{\dagger}_\ell\tau U_\ell}\bigg)\nonumber\\
&&\leq \frac{1}{N}\sum_{\ell\in S}\bigg(\relent{W\left(\Psi_{RC_0}\otimes \ketbra{0}_{E_0}\otimes \xi^{a:n}_D\right)W^{\dagger}\otimes\ketbra{0}_{\qbit}\otimes \sigma_{C_1E_1}}{\Psi_R\otimes \mu_{\brce_1}\otimes \xi^{1:n}_D}\nonumber\\
&& \hspace{2cm} - \relent{W\left(\Psi_{RC_0}\otimes \ketbra{0}_{E_0}\otimes \xi^{a:n}_D\right)W^{\dagger}\otimes\ketbra{0}_{\qbit}\otimes \sigma_{C_1E_1}}{\Tr_{\brce_2}\left(U^{\dagger}_\ell\tau U_\ell\right)}\bigg).
\end{eqnarray}
Now, we have
\begin{eqnarray}
\label{crossterms}
&&\Tr_{\brce_2}\left(U^{\dagger}_\ell\tau U_\ell\right)\nonumber\\
&&= \frac{1}{N}W\left(\Psi_{RC_0}\otimes \ketbra{0}_{E_0}\otimes \xi^{a:n}_D\right)W^{\dagger}\otimes \ketbra{0}_{\qbit}\otimes\sigma_{C_1E_1}\nonumber\\
&&+ \frac{1}{N}\sum_{m\in S, m\neq \ell}\Tr_{\brce_2}\left(U_{m-\ell}\left(W\left(\Psi_{RC_0}\otimes \ketbra{0}_{E_0}\otimes \xi^{a:n}_D\right)W^{\dagger}\otimes\ketbra{0}_{\qbit}\otimes \sigma_{C_1E_1}\otimes \mu_{\brce_2}\right)U_{m-\ell}^{\dagger}\right)\nonumber\\
&&\preceq \frac{1}{N}W\left(\Psi_{RC_0}\otimes \ketbra{0}_{E_0}\otimes \xi^{a:n}_D\right)W^{\dagger}\otimes \ketbra{0}_{\qbit}\otimes\sigma_{C_1E_1} + \frac{(1+15\delta)(N-1)}{N}\Psi_R\otimes \mu_{\brce_1}\otimes \xi^{1:n}_D.\nonumber\\
\end{eqnarray}
Moreover, using the relation $\Psi_{RC_0}\preceq 2^{k+1}\Psi_R\otimes \sigma_{C_0}$, Equation \ref{unitarybros} and Claim \ref{embezclose}, we conclude
\begin{eqnarray*}
&&W\left(\Psi_{RC_0}\otimes \ketbra{0}_{E_0}\otimes \xi^{a:n}_D\right)W^{\dagger}\otimes \ketbra{0}_{\qbit}\otimes\sigma_{C_1E_1}\\
&&\preceq 2^{k+1} \Psi_R\otimes W\left(\sigma_{C_0}\otimes \ketbra{0}_{E_0}\otimes \xi^{a:n}_D\right)W^{\dagger}\otimes \ketbra{0}_{\qbit}\otimes\sigma_{C_1E_1}\\
&&\preceq 2^{k+1}(1+15\delta) \Psi_R\otimes\sigma_{C_0E_0}\otimes\ketbra{0}_{\qbit}\otimes \sigma_{C_1E_1}\otimes \xi^{1:n}_D\\
&&\preceq 2^{k+2}(1+15\delta)\cdot \frac{|\brce|}{|\supp(\sigma_{CE})|^2} \Psi_R\otimes\mu_{\brce_1}\otimes \xi^{1:n}_D\\
&&\preceq 2^{k+2}(1+15\delta) \Psi_R\otimes\mu_{\brce_1}\otimes \xi^{1:n}_D.
\end{eqnarray*} 
Using this in Equation \ref{crossterms}, we conclude that
$$\Tr_{\brce_2}\left(U^{\dagger}_\ell\tau U_\ell\right) \preceq (1+15\delta)\cdot\left(1+\frac{2^{k+2}-1}{N}\right)\Psi_R\otimes\mu_{\brce_1}\otimes \xi^{1:n}_D.$$
Along with Equation \ref{elemeqagain}, this leads to 
$$\relent{\tau}{\Psi_R\otimes \mu_{\brce_1}\otimes \xi^{1:n}_D\otimes \mu_{\brce_2}}\leq 15\delta+\log\left(1+\frac{2^{k+2}-1}{N}\right).$$
This completes the proof.
\end{proof}

\subsection{Proofs in Section \ref{sec:posbased}}
\label{proofs:posbased}

\begin{proof}[\bf Proof of Theorem \ref{theo:posbased}]
Let $\Omega_{BC}$ be the operator such that $$\Tr\left(\Omega_{BC}\Psi_{BC}\right)\geq 1-\eps, \quad\Tr\left(\Omega_{BC}\Psi_B\otimes \mu_C\right) = 2^{-\dheps{\Psi_{BC}}{\Psi_B\otimes \mu_C}{\eps}}.$$ We have
\begin{eqnarray}
\label{bighypo}
\Tr(\Omega_{BC_0}\Psi_B\otimes \mu_{\brc_1}) &\leq& \frac{|\qbit||C|^2}{|\brc|} \Tr(\Omega_{BC_0}\Psi_B\otimes \mu_{C_0}\otimes\mu_{\qbit}\otimes \mu_{C_1}) \nonumber\\
&\leq& 2\cdot 2^{-\dheps{\Psi_{BC}}{\Psi_B\otimes \mu_C}{\eps}} = 2^{1-\dheps{\Psi_{BC}}{\Psi_B\otimes \mu_C}{\eps}}.
\end{eqnarray}
Let $\{\Lambda_{-1}, \Lambda_{\ell}\}_{\ell\in \cS}$ be the POVM as constructed in Fact \ref{haynag} using the operators $U_{\ell}\Omega_{BC_0}U^{\dagger}_{\ell}$. We have
\begin{eqnarray*}
\Tr\left((\id-\Lambda_\ell)\tau_\ell\right)&=& \Tr\left((\id-\Lambda_{\ell}) U_\ell\Psi_{BC_0}\otimes\ketbra{0}_{\qbit}\otimes\mu_{C_1}\otimes\mu_{\brc_2}U^{\dagger}_\ell\right)\\
&\overset{(a)}\leq& (1+c)\Tr\left((\id-U_{\ell}\Omega_{BC_0}U^{\dagger}_{\ell}) U_\ell\Psi_{BC_0}\otimes\ketbra{0}_{\qbit}\otimes\mu_{C_1}\otimes\mu_{\brc_2}U^{\dagger}_\ell\right)\\
&&+ (2+c+ c^{-1})\sum_{m\in \cS, m\neq \ell}\Tr\left((U_{m}\Omega_{BC_0}U^{\dagger}_{m}) U_\ell\Psi_{BC_0}\otimes\ketbra{0}_{\qbit}\otimes\mu_{C_1}\otimes\mu_{\brc_2}U^{\dagger}_\ell\right)\\
&=& (1+c)\Tr\left((\id-\Omega_{BC_0})\Psi_{BC_0}\otimes\ketbra{0}_{\qbit}\otimes\mu_{C_1}\otimes\mu_{\brc_2}\right)\\
&&+ (2+c+ c^{-1})\sum_{m\in \cS, m\neq \ell}\Tr\left(\Omega_{BC_0} U_{\ell-m}\Psi_{BC_0}\otimes\ketbra{0}_{\qbit}\otimes\mu_{C_1}\otimes\mu_{\brc_2}U^{\dagger}_{\ell-m}\right)\\
&\leq& (1+c)\eps + (2+c+ c^{-1})\sum_{m\in \cS, m\neq \ell}\Tr\left(\Omega_{BC_0} U_{\ell-m}\Psi_{BC_0}\otimes\ketbra{0}_{\qbit}\otimes\mu_{C_1}\otimes\mu_{\brc_2}U^{\dagger}_{\ell-m}\right).
\end{eqnarray*}
Above, $(a)$ uses Fact \ref{haynag}. From Lemma \ref{lem:Uprop2}, 
$$\Tr_{\brc_2}\left(U_{\ell-m}\Psi_{BC_0}\otimes\ketbra{0}_{\qbit}\otimes\mu_{C_1}\otimes\mu_{\brc_2}U^{\dagger}_{\ell-m}\right) = \Psi_B\otimes \mu_{\brc_1}.$$
Thus choosing $c= \frac{\delta}{\eps}$ and using Equation \ref{bighypo}, 
$$\Tr\left((\id-\Lambda_\ell)\tau_\ell\right) \leq \eps+\delta + \frac{4\eps}{\delta}|\cS|\Tr\left(\Omega_{BC_0} \Psi_{B}\otimes\mu_{\brc_1}\right) \leq \eps+\delta + \frac{4\eps}{\delta}|\cS|2^{1-\dheps{\Psi_{BC}}{\Psi_B\otimes \mu_C}{\eps}}\leq \eps+4\delta,$$
from the choice of $|\cS|$. This completes the proof.
\end{proof}

\begin{proof}[\bf Proof of Theorem \ref{cor:posbased}]
We will outline the main steps of the proof, which closely follow those of Theorem \ref{theo:posbased}. Let $\Omega_{BC}$ be the operator that satisfies
$$\Tr\left(\Omega_{BC}\Psi_{BC}\right)\geq 1-\eps, \quad\Tr\left(\Omega_{BC}\Psi_B\otimes \omega_C\right) = 2^{-\dheps{\Psi_{BC}}{\Psi_B\otimes \omega_C}{\eps}}.$$
From Fact \ref{nearbygood}, we have $$\sigma_C \preceq \frac{1}{1-\gamma}\omega_C \implies \Tr(\Omega_{BC}\Psi_B\otimes \sigma_C)\leq 2\cdot\Tr(\Omega_{BC}\Psi_B\otimes \omega_C)= 2^{1-\dheps{\Psi_{BC}}{\Psi_B\otimes \omega_C}{\eps}}.$$ 
Let $\{\Lambda_{-1},\Lambda_\ell\}_{\ell\in \cS}$ be the POVM constructed in Fact \ref{haynag} using the operators $\{U_{\ell}W\Omega_{BC_0}W^{\dagger}U^{\dagger}_{\ell}\}_{\ell\in \cS}$.
Rest of the calculation follows using Fact \ref{haynag}. The following claim is similar to Lemma \ref{lem:Uprop2}.
\begin{claim}
\label{embezprop2}
For any $m\in\left\{0,1,\ldots |\brce|-1\right\}$, it holds that $\Tr_{\brce_2}(\tau_m) \preceq (1+15\delta)\Psi_R\otimes \mu_{\brce_1}\otimes \xi^{1:n}_D$.
\end{claim}

\begin{proof}
We expand $\Psi_{RC_0}= \sum_{c,c'}\Psi^{(c,c')}_R\otimes \ket{c}\bra{c'}_{C_0}$. For convenience, set $b(c)= q(c)|C|/\gamma$. Recall that $$W\ket{c}_{C_0}\ket{0}_{E_0}\ket{k}_D = \ket{c}_{C_0}\ket{k\mmod{b(c)}}_{E_0}\ket{\lfloor k/b(c)\rfloor}_{D}.$$ Thus, 
\begin{eqnarray*}
&&\Tr_{\brce_2}\left(U_m\left(W\left(\ket{c}\bra{c'}_{C_0}\otimes \ketbra{0}_{E_0}\otimes \xi^{a:n}_D\right)W^{\dagger}\otimes \ketbra{0}_{\qbit}\otimes\sigma_{C_1E_1}\otimes \mu_{\brce_2}\right)U_m^{\dagger}\right)\\
&&= \sum_{k=a}^n \frac{1}{k}\Tr_{\brce_2}\left(U_m\left(W\left(\ket{c}\bra{c'}_{C_0}\otimes \ketbra{0}_{E_0}\otimes \ketbra{k}_D\right)W^{\dagger}\otimes\ketbra{0}_{\qbit}\otimes \sigma_{C_1E_1}\otimes \mu_{\brce_2}\right)U_m^{\dagger}\right)\\
&&=\sum_{k=a}^n \frac{1}{k}\Tr_{\brce_2}\bigg(U_m\left(\ket{c}\bra{c'}_{C_0}\otimes \ket{k\mmod{b(c)}}\bra{k\mmod{b(c')}}_{E_0}\otimes \ketbra{0}_{\qbit}\otimes\sigma_{C_1E_1}\otimes \mu_{\brce_2}\right)U_m^{\dagger}\\
&&\hspace{4cm}\otimes \ket{\lfloor k/b(c)\rfloor}\bra{\lfloor k/b(c')\rfloor}_D\bigg)
\end{eqnarray*}
As shown in Lemma \ref{lem:Uprop2}, 
\begin{eqnarray*}
\Tr_{\brce_2}\left(U_m\left(\ket{c}\bra{c'}_{C_0}\otimes \ket{k\mmod{b(c)}}\bra{k\mmod{b(c')}}_{E_0}\otimes\ketbra{0}_{\qbit}\otimes \sigma_{C_1E_1}\otimes\mu_{\brce_2}\right)U_m^{\dagger}\right)=\mu_{\brce_1}\cdot \delta_{c,c'}.
\end{eqnarray*}
 Hence, we conclude that
\begin{eqnarray*}
&&\Tr_{\brce_2}\left(U_m\left(W\left(\ket{c}\bra{c'}_{C_0}\otimes \ketbra{0}_{E_0}\otimes \xi^{a:n}_D\right)W^{\dagger}\otimes \ketbra{0}_{\qbit}\otimes\sigma_{C_1E_1}\otimes \mu_{\brce_2}\right)U_m^{\dagger}\right)\\
&&=\mu_{\brce_1}\otimes \sum_{k=a}^n \frac{1}{k}\ket{\lfloor k/b(c)\rfloor}\bra{\lfloor k/b(c)\rfloor}_D\cdot\delta_{c,c'}\\
&& \preceq (1+15\delta)\mu_{\brce_1}\otimes \xi^{1:n}_D\cdot \delta_{c,c'}, 
\end{eqnarray*}
where in the last operator inequality, we have used an argument similar to that used in Claim \ref{embezclose}. Thus,
\begin{eqnarray*}
\Tr_{\brce_2}\tau_m &=& \sum_{c,c'}\Psi_R^{(c,c')}\otimes \Tr_{\brce_2}\left(U_m\left(W\left(\ket{c}\bra{c'}_{C_0}\otimes \ketbra{0}_{E_0}\otimes \xi^{a:n}_D\right)W^{\dagger}\otimes\ketbra{0}_{\qbit}\otimes \sigma_{C_1E_1}\otimes \mu_{\brce_2}\right)U_m^{\dagger}\right)\\
&\preceq & (1+15\delta)\sum_{c}\Psi_R^{(c,c)} \otimes \mu_{\brce_1}\otimes \xi^{1:n}_D.  
\end{eqnarray*}
This completes the proof.
\end{proof}

 We require the following inequality for $m, \ell\in S$ with $m\neq \ell$.
\begin{eqnarray*}
&&\Tr\left(U_mW\Omega_{BC_0}W^{\dagger}U^{\dagger}_m\tau_\ell\right)\\
&&=\Tr\left(U_mW\Omega_{BC_0}W^{\dagger}U^{\dagger}_mU_\ell\left(W\left(\Psi_{BC_0}\otimes \ketbra{0}_{E_0}\otimes \xi^{a:n}_D\right)W^{\dagger}\otimes \ketbra{0}_{\qbit}\otimes\sigma_{C_1E_1}\otimes\mu_{\brce_2}\right)U^{\dagger}_\ell\right)\\
&&=\Tr\left(W\Omega_{BC_0}W^{\dagger}U_{\ell-m}\left(W\left(\Psi_{BC_0}\otimes \ketbra{0}_{E_0}\otimes \xi^{a:n}_D\right)W^{\dagger}\otimes \ketbra{0}_{\qbit}\otimes\sigma_{C_1E_1}\otimes\mu_{\brce_2}\right)U^{\dagger}_{\ell-m}\right)\\
&& \overset{(a)}\leq (1+15\delta)\Tr\left(\Omega_{BC_0}W^{\dagger}\left(\Psi_R\otimes \mu_{\brce_1}\otimes\xi^{1:n}_D\right)W\right)\\
&&\leq 2\cdot\Tr\left(\Omega_{BC_0}W^{\dagger}\left(\Psi_R\otimes \mu_{\brce_1}\otimes\xi^{1:n}_D\right)W\right).
\end{eqnarray*}
Above, $(a)$ follows from Claim \ref{embezprop2}. Now, we use the fact that $$\mu_{F_1} \preceq \frac{|\qbit||\supp(\sigma_{CE})|^2}{|F|}\mu_{\qbit}\otimes\sigma_{C_0E_0}\otimes \sigma_{C_1E_1} \preceq 2\cdot\mu_{\qbit}\otimes\sigma_{C_0E_0}\otimes \sigma_{C_1E_1}.$$
Thus,
\begin{eqnarray*}
&&\Tr\left(U_mW\Omega_{BC_0}W^{\dagger}U^{\dagger}_m\tau_\ell\right)\\
&&\leq 2\cdot\Tr\left(\Omega_{BC_0}W^{\dagger}\left(\Psi_R\otimes \mu_{\brce_1}\otimes\xi^{1:n}_D\right)W\right)\\
&&\leq 4\cdot\Tr\left(\Omega_{BC_0}W^{\dagger}\left(\Psi_R\otimes\mu_{\qbit}\otimes \mu_{C_0E_0}\otimes\mu_{C_1E_1}\otimes\xi^{1:n}_D\right)W\right)\\
&&= 4\cdot\Tr\left(\Omega_{BC_0}W^{\dagger}\left(\Psi_R\otimes \mu_{C_0E_0}\otimes\xi^{1:n}_D\right)W\right).
\end{eqnarray*}
Finally, we use Claim \ref{unembezclose} to conclude that
\begin{eqnarray*}
&&\Tr\left(U_mW\Omega_{BC_0}W^{\dagger}U^{\dagger}_m\tau_\ell\right)\\
&&\leq 4\cdot\Tr\left(\Omega_{BC_0}W^{\dagger}\left(\Psi_R\otimes \mu_{C_0E_0}\otimes\xi^{1:n}_D\right)W\right)\\
&&\leq 16\cdot \Tr\left(\Omega_{BC_0}\Psi_R\otimes \sigma_{C_0}\otimes \ketbra{0}_{E_0}\otimes\xi^{1:|D|}_D\right)\\
&&= 16\cdot\Tr\left(\Omega_{BC_0}\Psi_R\otimes \sigma_{C_0}\right) \leq 2^{5-\dheps{\Psi_{BC}}{\Psi_B\otimes \omega_C}{\eps}},
\end{eqnarray*}
where we use the fact that $\Omega_{BC}$ only acts in the support of $\Psi_R\otimes\sigma_{C_0}$. This completes the proof.
\end{proof}

\subsection{Proofs in Section \ref{sec:apps}}
\label{proofs:apps}

  For the ease of presentation, we will represent the relation $\Pur(\ketbra{\psi}, \ket{\phi})\leq \eps$ between two pure states $\ket{\psi}, \ket{\phi}$ as $\ket{\psi} \overset{\eps}\approx \ket{\phi}$.

\begin{proof}[\bf Proof of Theorem \ref{optchancode}]
Without loss of generality, we can assume that $\ket{\Psi}_{AC}$ is the canonical purification of $\Psi_A$, by applying a local unitary on register $C$ which does not change the hypothesis testing relative entropy. Let $\Psi_{BC}\defeq \chnl{\Psi_{AC}}$. From Fact \ref{nearbygood}, there exists a quantum state $\sigma_C$ such that the eigenvalues of $\sigma_C$ are integer multiples of $\frac{\gamma}{|C|}$ and $$ \sigma_C\preceq \frac{1}{1-\gamma}\Psi_C\implies \Pur(\Psi_C,\sigma_C) \leq \sqrt{\gamma}.$$
Let $\ket{\sigma}_{AC}$ be the canonical purification of $\sigma_C$ and $\sigma_{BC}\defeq \chnl{\sigma_{AC}}$. Using Fact \ref{canonicalfid}, $$\sigma_A\preceq \frac{1}{1-\gamma}\Psi_A, \quad\Pur(\sigma_{AC},\Psi_{AC}) \leq  2\sqrt{\Pur(\sigma_C,\Psi_C)}\leq 2\gamma^{1/4}$$ and using Fact \ref{fact:dataprocessing}, $$\sigma_B=\chnl{\sigma_A}\preceq \frac{1}{1-\gamma}\chnl{\Psi_A}=\frac{1}{1-\gamma}\Psi_B,\quad \Pur(\sigma_{BC},\Psi_{BC})\leq \Pur(\sigma_{AC},\Psi_{AC})\leq 2\gamma^{1/4}.$$  Let $\Omega_{BC}$ be the optimum operator in the definition of $\dheps{\Psi_{BC}}{\Psi_B\otimes\Psi_C}{\eps}$. From Fact \ref{closestatesmeasurement}
\begin{equation}
\label{eq:sigmadh}
\Tr\left(\Omega_{BC}\sigma_{BC}\right)\geq 1-\eps-4\gamma^{1/4}, \quad \Tr\left(\Omega_{BC}\sigma_B\otimes\sigma_C\right)\leq \frac{1}{(1-\gamma)^2}\Tr\left(\Omega_{BC}\Psi_B\otimes\Psi_C\right)\leq 2^{2-\dheps{\Psi_{BC}}{\Psi_B\otimes\Psi_C}{\eps}}.
\end{equation}
We expand $\ket{\sigma}_{AC}= \sum_c\sqrt{q(c)}\ket{c}_A\ket{c}_C$. Let $E$ be the register and $\sigma_{CE}$ be the quantum state as obtained in Definition \ref{broextend}. It holds that $|E|\leq \frac{|A|}{\gamma}$. Consider the following purification of $\sigma_{CE}$, which is maximally entangled.
$$\ket{\sigma'}_{ACE'E}\defeq \sqrt{\frac{\gamma}{|C|}}\sum_{c,e: e\leq \frac{\gamma q(c)}{|C|}}\ket{c,e}_{AE'}\ket{c,e}_{CE}.$$
Let $a\defeq \frac{|C|}{\gamma\cdot\delta}, n= a^{\frac{1}{\delta}}$ and register $D$ satisfy $|D|= n\left(\frac{|C|}{\gamma}+1\right)$. This ensures that Claims \ref{embezclose}, \ref{unembezclose} and \ref{purembezzle} apply to register $E$. From Definition \ref{def:embez}, $\ket{\xi^{a:n}}_{D'D}$ is the canonical purification of $\xi^{a:n}_D$ with $D'\equiv D$. Given the unitary $W \defeq \sum_c \ketbra{c}\otimes W_{\frac{q(c)|C|}{\gamma}}$ from Definition \ref{unitaryflat}, let $\wal\defeq W_{AE'D'}$ and $\wbob\defeq W_{CED}$. Using Definition \ref{paulis}, we obtain $|\supp(\sigma_{CE})|^2=\left(\frac{|C|}{\gamma}\right)^2$ unitaries $V_x:\supp(\sigma_{CE})\rightarrow \supp(\sigma_{CE})$. From Claim \ref{purembezzle}, we have 
$$\left(\wal\otimes \wbob\right)\ket{\sigma}_{AC}\otimes\ket{\xi^{a:n}}_{D'D}\otimes \ket{0,0}_{E'E} \overset{5\sqrt{\delta}}\approx  \ket{\sigma'}_{ACE'E}\otimes \ket{\xi^{1:n}}_{D'D}.$$
Since $V_x$ acts in $\supp(\sigma_{CE})$, Fact \ref{transposetrick} ensures that there exists a unitary $V^T_x:\supp(\sigma'_{AE'})\rightarrow\supp(\sigma'_{AE'})$ such that $(V^T_x\otimes\id)\ket{\sigma'}_{ACE'E}= (\id\otimes V_x)\ket{\sigma'}_{ACE'E}$. Thus, we obtain 
$$\left(V^T_x\wal\otimes \wbob\right)\ket{\sigma}_{AC}\otimes\ket{\xi^{a:n}}_{D'D}\otimes \ket{0,0}_{E'E} \overset{5\sqrt{\delta}}\approx (\id\otimes V_x) \ket{\sigma'}_{ACE'E}\otimes \ket{\xi^{1:n}}_{D'D}.$$
By triangle inequality for purified distance (Fact \ref{fact:trianglepurified}), these equations lead to 
\begin{equation}
\label{eq:transposeclose}
\left((\wal)^{\dagger}V^T_x\wal\otimes \wbob\right)\ket{\sigma}_{AC}\otimes\ket{\xi^{a:n}}_{D'D}\otimes \ket{0,0}_{E'E} \overset{10\sqrt{\delta}}\approx (\id\otimes V_x\wbob)\ket{\sigma}_{AC}\otimes\ket{\xi^{a:n}}_{D'D}\otimes \ket{0,0}_{E'E}.
\end{equation}
Introduce registers $X_1, X_2$ where  $|X_1|=|X_2|=\left(\frac{|C|}{\gamma}\right)^2$. Since $\left(\frac{|C|}{\gamma}\right)^2$ is a prime power, Definition \ref{def:pairwiseunit} gives a family of functions $\{f_m: \cX\times \cX\rightarrow \cX\}$ and a collection of unitaries $$V^{(m)}= \sum_{x_1,x_2}V_{f_m(x_1,x_2)}\otimes\ketbra{x_1, x_2}_{X_1X_2}.$$
 Let $\{\Lambda_{-1},\Lambda_1,\ldots \Lambda_{2^R}\}$ be POVM as defined in Fact \ref{haynag} using the operators $\{(V^{(m)}\wbob)\Omega_{BC}(V^{(m)}\wbob)^{\dagger}\}_{m=1}^{2^R}$.

\vspace{0.1in}

\noindent{\bf Shared resources:}  Alice and Bob share the state $\ket{\sigma}_{AC}\ket{\xi^{a:n}}_{D'D}\ket{0,0}_{E'E}$. They also possess $\mu_{X_1 X_{2}}$ in shared registers $X_1X_2$. Thus, the number of qubits of shared entanglement is $$\log|C| + \log|D| \leq \log n + 2\log\frac{|C|}{\gamma} = \frac{1}{\delta}\log\frac{|C|}{\gamma\cdot\delta}+2\log\frac{|C|}{\gamma}=\frac{1}{\delta}\log\frac{|A|}{\gamma\cdot\delta}+2\log\frac{|A|}{\gamma}.$$ 
\vspace{0.1in}

\noindent{\bf Encoding:} To send the message $m\in \{1,2,\ldots 2^R\}$, Alice applies the unitary $$\sum_{x_1,x_2}(\wal)^{\dagger}V^T_{f_m(x_1,x_2)}\wal\otimes \ketbra{x_1,x_2}_{X_1X_2}$$ on her registers. She then sends the register $A$ through the channel.

\vspace{0.1in}

\noindent{\bf Decoding:} Bob applies the unitary $\wbob$ on his registers. He applies the POVM $\{\Lambda_{-1},\Lambda_1,\ldots \Lambda_{2^R}\}$ and outputs $m'$ upon obtaining the outcome $\Lambda_{m'}$.

\vspace{0.1in}

\noindent{\bf Error analysis:}  Let $\theta'_m$ be the quantum state on Bob's registers just after Alice's transmission through the channel. Define the following quantum state: 
\begin{eqnarray*}
\theta_m&\defeq& \frac{1}{|X_1|^2}\sum_{x_1,x_{2}}\ketbra{x_1,x_2}_{X_1X_2}\otimes(V_{f_m(x_1,x_2)}\wbob)\left(\sigma_{BC}\otimes\xi^{a:n}_{D}\otimes \ketbra{0}_E\right)(V_{f_m(x_1,x_2)}\wbob)^{\dagger}.
\end{eqnarray*}
From Equation \ref{eq:transposeclose}, we have $\Pur(\theta_m, \theta'_m)\leq 10\sqrt{\delta}$. Thus, from Fact \ref{closestatesmeasurement}, 
$$\Pr[M'\neq m] = \Tr\left((1-\Lambda_m)\theta'_m\right) \leq \Tr\left((1-\Lambda_m)\theta_m\right) + 2\Pur(\theta_m, \theta'_m) \leq 20\sqrt{\delta} +  \Tr\left((1-\Lambda_m)\theta_m\right).$$
Applying Fact \ref{haynag}, we conclude 
\begin{eqnarray}
\label{errboundopt}
\Pr[M'\neq m] &=& 20\sqrt{\delta} +\Tr\left((1-\Lambda_m)\theta_m\right)\nonumber\\
&\leq& 20\sqrt{\delta} +(1+c)\left(1-\Tr\left((V^{(m)}\wbob)\Omega_{BC}(V^{(m)}\wbob)^{\dagger}\theta_m\right)\right)\nonumber\\
&&+ (2+c+c^{-1})\sum_{m'\neq m}\Tr\left((V^{(m')}\wbob)\Omega_{BC}(V^{(m')}\wbob)^{\dagger}\theta_m\right)\nonumber\\
&=& 20\sqrt{\delta} +(1+c)\left(1-\Tr\left(\Omega_{BC}(V^{(m)}\wbob)^{\dagger}\theta_m(V^{(m)}\wbob)\right)\right)\nonumber\\
&&+ (2+c+c^{-1})\sum_{m'\neq m}\Tr\left((V^{(m')}\wbob)\Omega_{BC}(V^{(m')}\wbob)^{\dagger}\theta_m\right).
\end{eqnarray}
Since
\begin{eqnarray*}
&&(V^{(m)}\wbob)^{\dagger}\theta_m(V^{(m)}\wbob)\\
&&= \frac{1}{|X_1|^2}\sum_{x_1,x_2}\ketbra{x_1,x_{2}}_{X_1X_2}\otimes\left(\sigma_{BC}\otimes\xi^{a:n}_{D}\otimes \ketbra{0}_E\right),
\end{eqnarray*}
from Equation \ref{eq:sigmadh}, we have 
\begin{equation}
\label{eq:rotsigfirstterm}
\Tr\left(\Omega_{BC}(V^{(m)}\wbob)^{\dagger}\theta_m(V^{(m)}\wbob)\right) = \Tr(\Omega_{BC}\sigma_{BC})\geq 1-\eps-4\gamma^{1/4}.
\end{equation}
For $m'\neq m$, consider
\begin{eqnarray*}
&&\Tr\left((V^{(m')}\wbob)\Omega_{BC}(V^{(m')}\wbob)^{\dagger}\theta_m\right) =\frac{1}{|X_1|^2}\sum_{x_1, x_{2}}\\
&&\Tr\left(\Omega_{BC}(V_{f_{m'}(x_1,x_2)}\wbob)^{\dagger}(V_{f_{m}(x_1,x_2)}\wbob)\left(\sigma_{BC}\otimes\xi^{a:n}_{D}\otimes \ketbra{0}_E\right)(V_{f_{m}(x_1,x_2)}\wbob)^{\dagger}(V_{f_{m'}(x_1,x_2)}\wbob)\right)\\
&&= \frac{1}{|\cX|}\sum_{x}\Tr\bigg(\Omega_{BC}(V_{x}\wbob)^{\dagger}\bigg(\frac{1}{|\cX|}\sum_{x'}V_{x'}\wbob\left(\sigma_{BC}\otimes\xi^{a:n}_{D}\otimes  \ketbra{0}_E\right)(V_{x'}\wbob)^{\dagger}\bigg)(V_{x}\wbob)\bigg),\\
\end{eqnarray*}
where we have used Definition \ref{def:pairwiseunit} to introduce variables $x,x'$ in a manner similar to Equation \ref{eq:changexxprime}. From Equation \ref{embezcrossvanish}, we have 
$$\frac{1}{|\cX|}\sum_{x'}V_{x'}\wbob\left(\sigma_{BC}\otimes\xi^{a:n}_{D}\otimes  \ketbra{0}_E\right)(V_{x'}\wbob)^{\dagger}\preceq (1+15\delta)\sigma_B\otimes \sigma_{CE}\otimes \xi^{1:n}.$$
Thus, 
\begin{eqnarray*}
&&\Tr\left((V^{(m')}\wbob)\Omega_{BC}(V^{(m')}\wbob)^{\dagger}\theta_m\right)\\
&&\leq \frac{(1+15\delta)}{|\cX|}\sum_{x}\Tr\bigg(\Omega_{BC}(V_{x}\wbob)^{\dagger}\bigg(\sigma_B\otimes \sigma_{CE}\otimes \xi^{1:n}\bigg)(V_{x}\wbob)\bigg)\\
&&\overset{(a)}=(1+15\delta)\Tr\bigg(\Omega_{BC}(\wbob)^{\dagger}\bigg(\sigma_B\otimes \sigma_{CE}\otimes \xi^{1:n}\bigg)(\wbob)\bigg)\\
&& \overset{(b)}\leq 4(1+15\delta)\Tr\bigg(\Omega_{BC}\bigg(\sigma_B\otimes \sigma_{C}\otimes \xi^{1:|D|}\otimes \ketbra{0}_E\bigg)\bigg)\\
&& \leq 8\cdot  \Tr\bigg(\Omega_{BC}\sigma_B\otimes \sigma_{C}\bigg)\overset{(c)}\leq 2^{5- \dheps{\Psi_{BC}}{\Psi_B\otimes\Psi_C}{\eps}}.
\end{eqnarray*}
where $(a)$ uses the fact that $V_x^{\dagger}\sigma_{CE}V_x=\sigma_{CE}$, $(b)$ uses Claim \ref{unembezclose} and $(c)$ uses Equation \ref{eq:sigmadh}. Using it with Equation \ref{eq:rotsigfirstterm} and Equation \ref{errboundopt}, we conclude
$$\Pr[M'\neq m] \leq 20\sqrt{\delta} + (1+c)(\eps+4\gamma^{1/4})+ \frac{4}{c}2^{R + 5- \dheps{\Psi_{BC}}{\Psi_B\otimes\Psi_C}{\eps}}.$$
Setting $c=\frac{\delta'}{\eps+4\gamma^{1/4}}$ and from the choice of $R$, the proof concludes. 
\end{proof}

\end{document}